\PassOptionsToPackage{unicode}{hyperref}
\PassOptionsToPackage{hyphens}{url}
\PassOptionsToPackage{dvipsnames,svgnames,x11names}{xcolor}
\documentclass[12pt]{article}

\usepackage{soul} 
\usepackage{comment}
\usepackage{multirow}
\usepackage{xcolor}
\usepackage[normalem]{ulem}

\usepackage{setspace}
\let\oldfootnote\footnote
\renewcommand{\footnote}[1]{\oldfootnote{\setstretch{1}#1}}

\newtheorem{definition}{Definition}
\newtheorem{theorem}{Theorem}
\newtheorem{lemma}{Lemma}
\newtheorem{conjecture}{Conjecture}

\newtheorem{example}{Example}
\newtheorem{proof}{Proof}

\usepackage{amsmath,amssymb}
\usepackage{iftex}
\ifPDFTeX
  \usepackage[T1]{fontenc}
  \usepackage[utf8]{inputenc}
  \usepackage{textcomp} 
\else 
  \usepackage{unicode-math}
  \defaultfontfeatures{Scale=MatchLowercase}
  \defaultfontfeatures[\rmfamily]{Ligatures=TeX,Scale=1}
\fi
\usepackage{lmodern}
\ifPDFTeX\else  
\fi
\IfFileExists{upquote.sty}{\usepackage{upquote}}{}
\IfFileExists{microtype.sty}{
  \usepackage[]{microtype}
  \UseMicrotypeSet[protrusion]{basicmath} 
}{}
\makeatletter
\@ifundefined{KOMAClassName}{
  \IfFileExists{parskip.sty}{%
    \usepackage{parskip}
  }{
    \setlength{\parindent}{0pt}
    \setlength{\parskip}{6pt plus 2pt minus 1pt}}
}{
  \KOMAoptions{parskip=half}}
\makeatother
\usepackage{xcolor}
\makeatletter
\ifx\paragraph\undefined\else
  \let\oldparagraph\paragraph
  \renewcommand{\paragraph}{
    \@ifstar
      \xxxParagraphStar
      \xxxParagraphNoStar
  }
  \newcommand{\xxxParagraphStar}[1]{\oldparagraph*{#1}\mbox{}}
  \newcommand{\xxxParagraphNoStar}[1]{\oldparagraph{#1}\mbox{}}
\fi
\ifx\subparagraph\undefined\else
  \let\oldsubparagraph\subparagraph
  \renewcommand{\subparagraph}{
    \@ifstar
      \xxxSubParagraphStar
      \xxxSubParagraphNoStar
  }
  \newcommand{\xxxSubParagraphStar}[1]{\oldsubparagraph*{#1}\mbox{}}
  \newcommand{\xxxSubParagraphNoStar}[1]{\oldsubparagraph{#1}\mbox{}}
\fi
\makeatother

\usepackage{longtable,booktabs,array}
\usepackage{calc} 
\usepackage{etoolbox}
\makeatletter
\patchcmd\longtable{\par}{\if@noskipsec\mbox{}\fi\par}{}{}
\makeatother
\IfFileExists{footnotehyper.sty}{\usepackage{footnotehyper}}{\usepackage{footnote}}
\makesavenoteenv{longtable}
\usepackage{graphicx}
\makeatletter
\def\maxwidth{\ifdim\Gin@nat@width>\linewidth\linewidth\else\Gin@nat@width\fi}
\def\maxheight{\ifdim\Gin@nat@height>\textheight\textheight\else\Gin@nat@height\fi}
\makeatother
\setkeys{Gin}{width=\maxwidth,height=\maxheight,keepaspectratio}
\makeatletter
\def\fps@figure{htbp}
\makeatother

\makeatletter
\@ifpackageloaded{caption}{}{\usepackage{caption}}
\AtBeginDocument{%
\ifdefined\contentsname
  \renewcommand*\contentsname{Table of contents}
\else
  \newcommand\contentsname{Table of contents}
\fi
\ifdefined\listfigurename
  \renewcommand*\listfigurename{List of Figures}
\else
  \newcommand\listfigurename{List of Figures}
\fi
\ifdefined\listtablename
  \renewcommand*\listtablename{List of Tables}
\else
  \newcommand\listtablename{List of Tables}
\fi
\ifdefined\figurename
  \renewcommand*\figurename{Figure}
\else
  \newcommand\figurename{Figure}
\fi
\ifdefined\tablename
  \renewcommand*\tablename{Table}
\else
  \newcommand\tablename{Table}
\fi
}
\@ifpackageloaded{float}{}{\usepackage{float}}
\floatstyle{ruled}
\@ifundefined{c@chapter}{\newfloat{codelisting}{h}{lop}}{\newfloat{codelisting}{h}{lop}[chapter]}
\floatname{codelisting}{Listing}

\makeatother
\makeatletter
\@ifpackageloaded{caption}{}{\usepackage{caption}}
\@ifpackageloaded{subcaption}{}{\usepackage{subcaption}}
\makeatother

\ifLuaTeX
  \usepackage{selnolig}  
\fi
\usepackage[]{natbib}
\usepackage{bookmark}

\IfFileExists{xurl.sty}{\usepackage{xurl}}{} 
\hypersetup{
  pdftitle={Title},
  pdfauthor={Author 1; Author 2; Author3},
  pdfkeywords={3 to 6 keywords, that do not appear in the title},
  colorlinks=true,
  linkcolor={blue},
  filecolor={Maroon},
  citecolor={Blue},
  urlcolor={Blue},
  pdfcreator={LaTeX via pandoc}}

\newcommand{\anon}{1}

\begin{document}

\def\spacingset#1{\renewcommand{\baselinestretch}%
{#1}\small\normalsize} \spacingset{1}


\if1\anon
{
  \title{\bf Constrained Gaussian Random Fields with Continuous Linear Boundary Restrictions for Physics-informed Modeling of States}
  \author{Yue Ma\thanks{The authors would like to acknowledge support from the U.S. National Science Foundation Engineering Research Center for Hybrid Autonomous Manufacturing Moving from Evolution to Revolution (ERC‐HAMMER) under Award Number EEC-2133630.}
    \hspace{.2cm}\\
    \normalsize Department of Statistics, The Ohio State University\\
    Oksana A. Chkrebtii$^{*\dagger}$ \\
    \normalsize Department of Statistics, The Ohio State University\\
    Stephen R. Niezgoda$^{*\dagger}$ \\
    \normalsize Department of Materials Science and Engineering, The Ohio State University}
  \maketitle
} \fi

\if0\anon
{
  \title{\bf Constrained Gaussian Random Fields with Continuous Linear Boundary Restrictions for Physics-informed Modeling of States}
  \maketitle
} \fi

\bigskip
\begin{abstract}
Boundary constraints in physical, environmental and engineering models restrict smooth states such as temperature to follow known physical laws at the edges of their spatio-temporal domains. Examples include fixed-state or fixed-derivative (insulated) boundary conditions, and constraints that relate the state and the derivatives, such as in models of heat transfer. Despite their flexibility as prior models over system states, Gaussian random fields do not in general enable exact enforcement of such constraints. This work develops a new general framework for constructing linearly boundary-constrained Gaussian random fields from unconstrained Gaussian random fields over multi-dimensional domains. This new class of models provides flexible priors for modeling smooth states with known physical mechanisms acting at the domain boundaries without requiring approximation. Simulation studies illustrate how such physics-informed probability models yield improved predictive performance and more realistic uncertainty quantification in applications including probabilistic numerics, data-driven discovery of dynamical systems, and boundary-constrained state estimation, as compared to unconstrained alternatives. The framework is then applied to the problem of recovering a displacement field from experimental tensile testing data, demonstrating its ability to integrate known boundary constraints within the inferential procedure.
\end{abstract}

\noindent%
\vfill

\spacingset{1.8} 

\section{Introduction}\label{sec-intro}

In the physical, environmental, and engineering sciences, where modeling smooth states is of interest, there is often complete information about the states at the domain boundaries. For example, when modeling temperature, insulation at the boundary leads to fixed-derivative constraints, while heat transfer results in constraints on linear combinations of the state and its derivatives \citep[e.g.,][]{Hahn2012HeatBC, Bird2006ChemistryBC}. Another example from materials science consists of modeling the displacement field along a material sample to study its response to pulling forces. A typical experimental set-up is illustrated in the left-hand panel of Figure \ref{fig:intro1} and shows that a statistical model should account for the known uniform motion at the upper boundary and zero displacement at the lower boundary of the material sample \citep[e.g.,][]{Hosford2005TensileBC, Arfken2011MathPhysicsBC}. Beyond physical restrictions, in the setting of computer model calibration, a fixed-state constraint over the response surface results when a complex mathematical model has an analytical solution for specific parameter combinations \citep{Vernon2019JUQ, Ye2022MultiFidelity}. Consequently, a prior model that enforces these constraints results in a more accurate surrogate model. Another important application is in constructing scalable approaches for fitting nonstationary models to large datasets, where some approaches \citep[e.g.,][]{Gramacy2015LocalGP,Huang2025StickGP} fit local models independently over disjoint domain partitions, necessitating fixed-state constraints at partition boundaries to ensure continuity of the underlying smooth states. In all these settings, flexible boundary-constrained priors are needed for improved predictive performance and uncertainty quantification given limited measurements.
\begin{figure}[t]
    \centering
    \fbox{\includegraphics[width=0.23\linewidth]{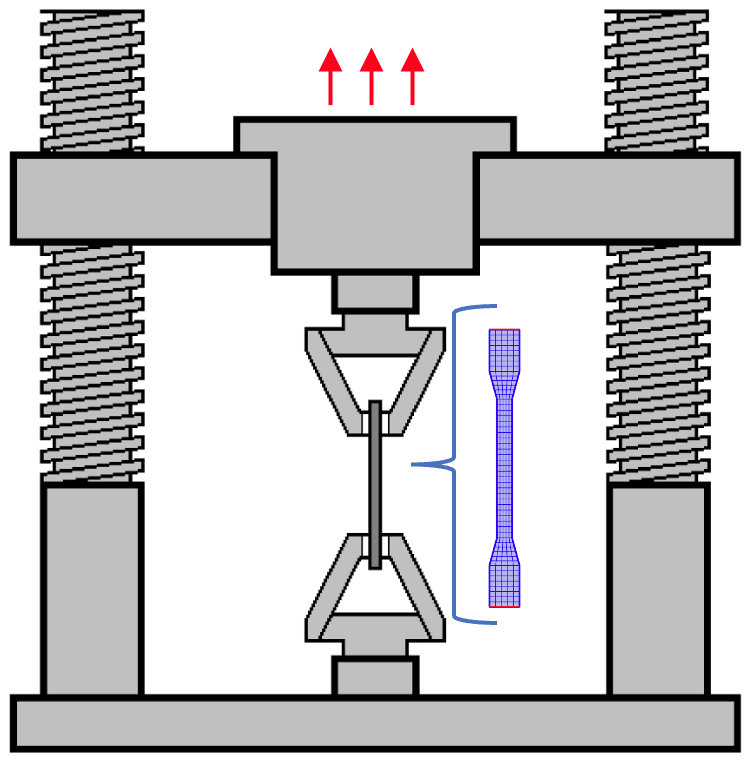}
    \includegraphics[width=0.063\linewidth]{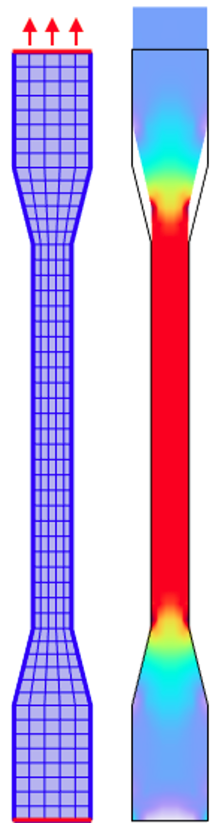}}\hspace{-0.04in}
    \fbox{\includegraphics[width=0.36\linewidth]{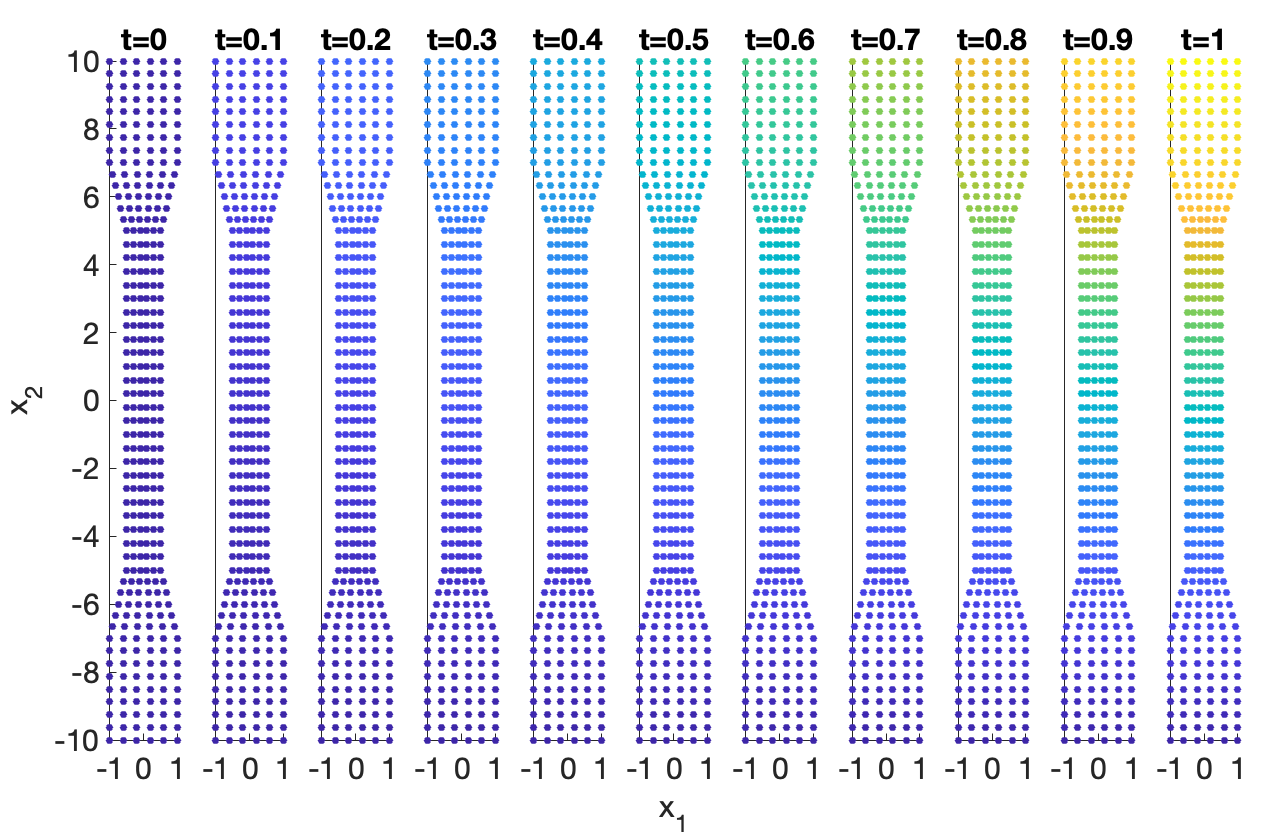}}\hspace{-0.04in}
    \fbox{\includegraphics[width=0.206\linewidth]{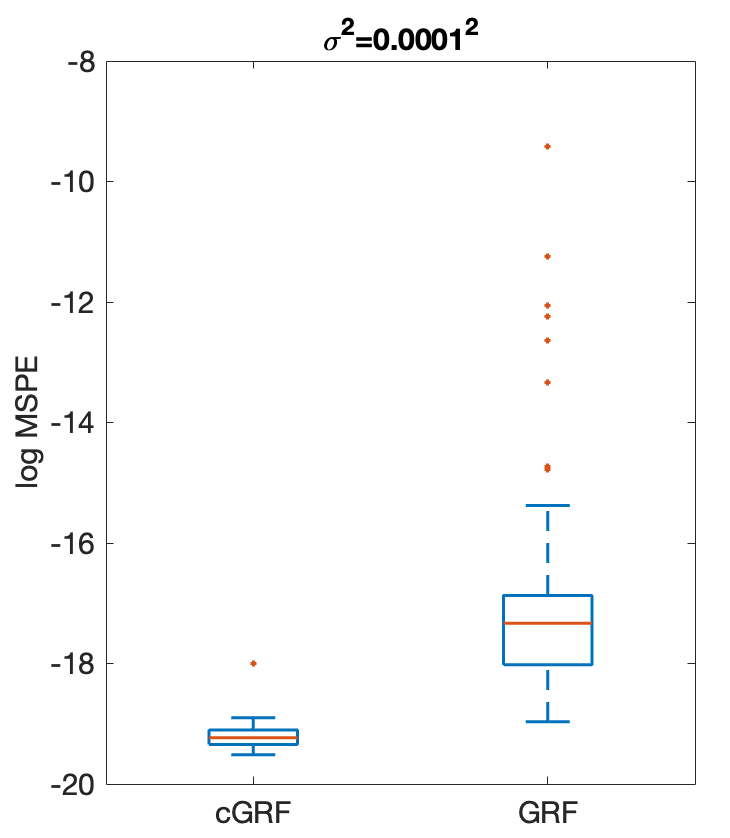}
    \includegraphics[width=0.057\linewidth]{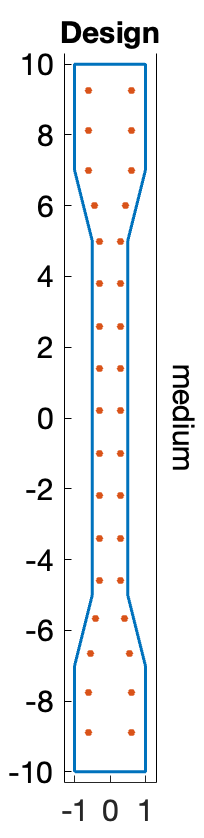}}
    \caption{
    Left panel: illustration of tensile strength test and corresponding constraints at two parallel boundaries (red horizontal lines).
    Center panel: simulated displacement over time from left to right; lighter colors represent greater displacement. 
    Right panel: box plots comparing $logMSPE$ under a constrained vs unconstrained GRF prior on the displacement field.}
    \label{fig:intro1}
\end{figure}

Existing approaches for constructing Gaussian random fields (GRFs) with boundary constraints are limited to specific covariance structures, hyperparameter choices, and constraint types on $1$-dimensional intervals or $2$-dimensional rectangular domains \citep{Ding2019BdryGPAN, Chkrebtii2013Thesis, Gasbarra2007GB, Ye2022MultiFidelity, Tan2018BI, Dalton2024BoundaryConstrainedGP, Vernon2019JUQ}. When extending beyond these cases, a common strategy is to rely on approximations, giving up either deriving analytical expressions for mean and covariance functions, or enforcing continuous constraints at multi-dimensional boundaries \citep{Solin2019Know, Wang2016Shape}. However, as will be illustrated in Section \ref{sec-app}, it is often desirable to maintain the flexibility of constrained GRF priors without relying on approximations. 
We propose a new framework to construct GRFs with continuous linear boundary restrictions over multi-dimensional domains (as illustrated in Figure \ref{fig:intro2}) by transforming appropriately smooth unconstrained GRFs. Expressions for the mean and covariance function of these restricted models are available in closed form, facilitating computation. 

\begin{figure}[t]
    \centering
    \fbox{\includegraphics[width=4.65in,height=\textheight]{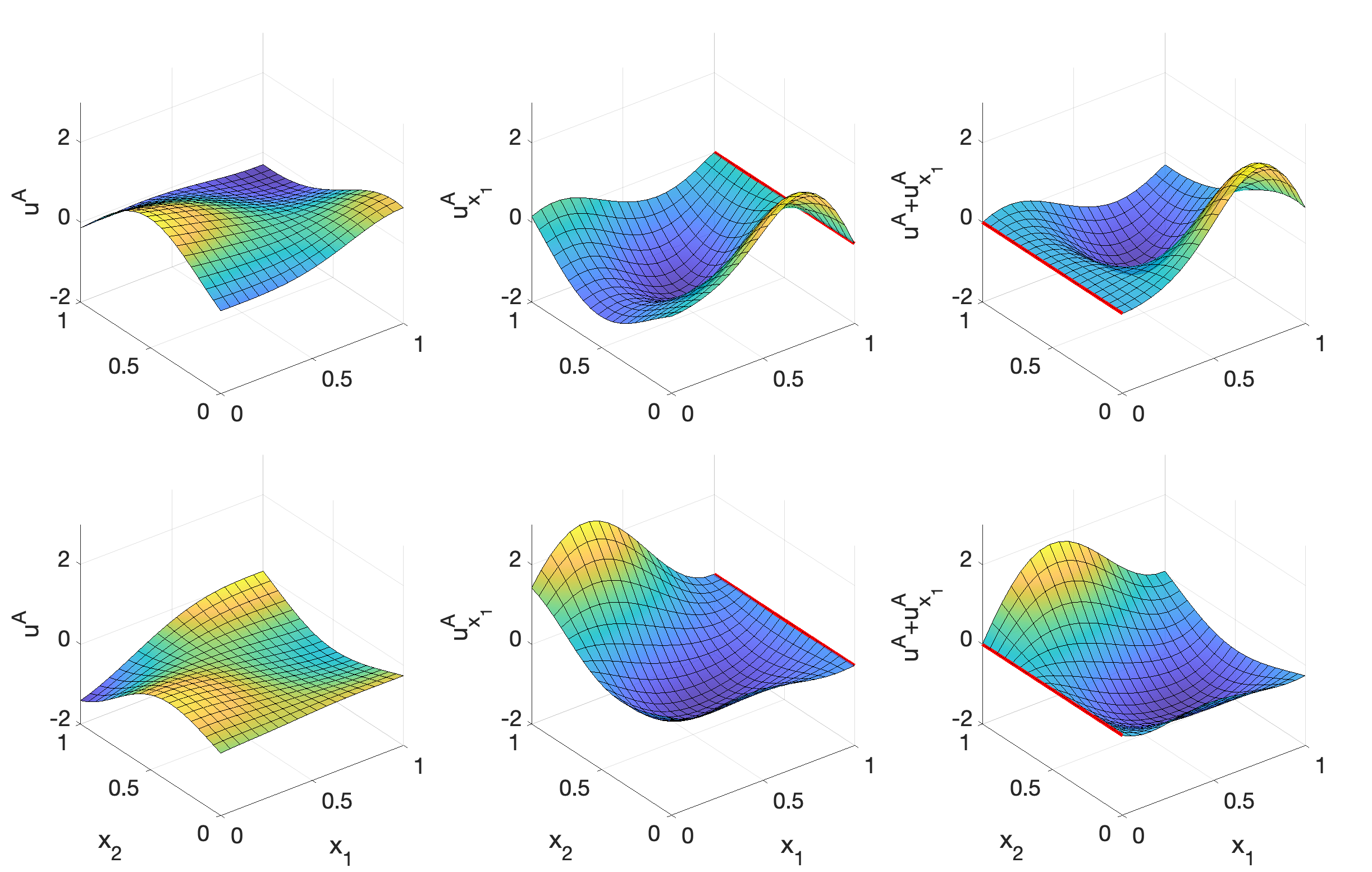}}\hspace{-0.04in}
    \fbox{\includegraphics[width=1.5in,height=\textheight]{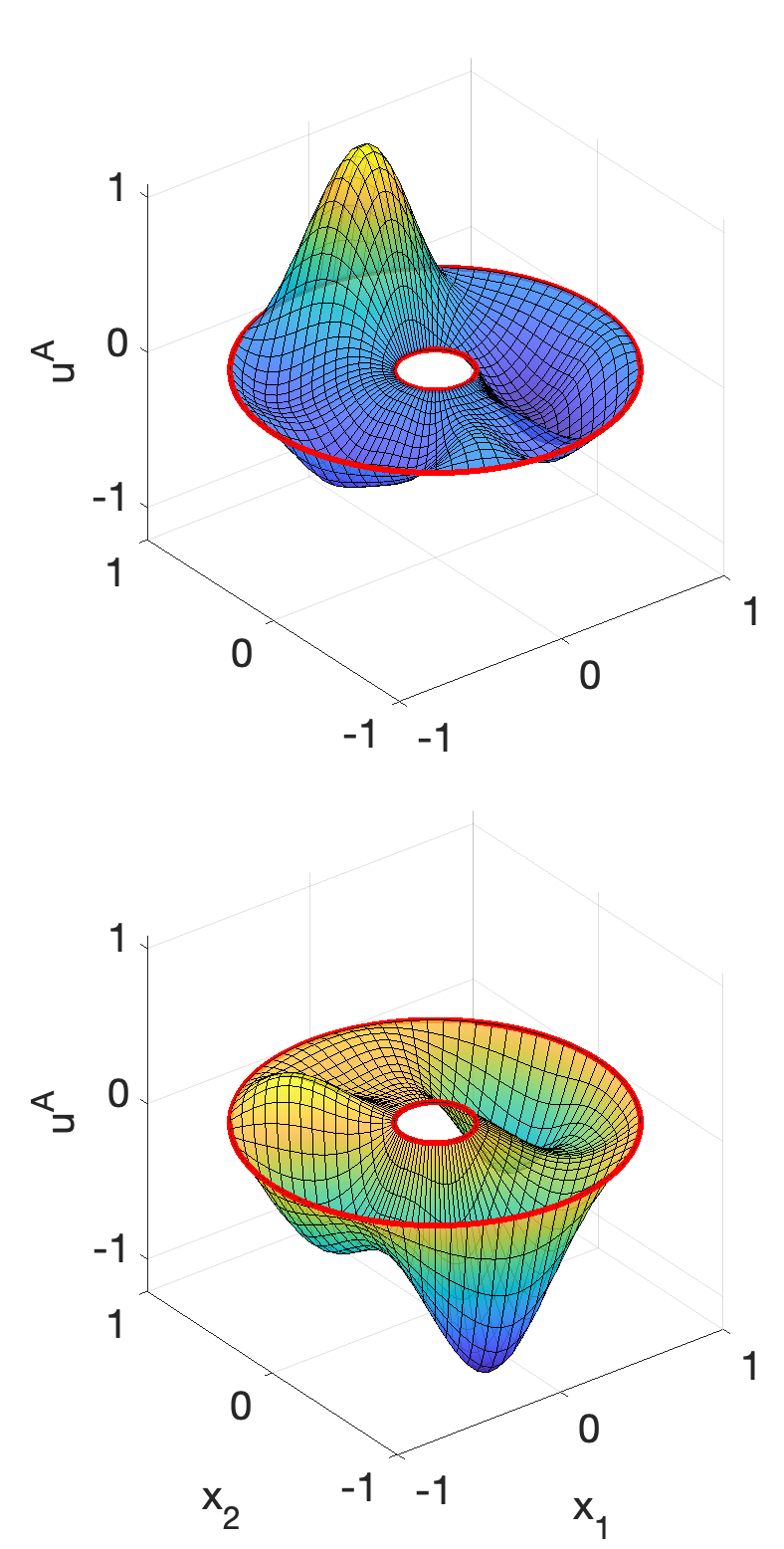}}
    \caption{
    Pairs of realizations (top and bottom rows) from the marginal cGRF model with boundary constraints shown in red.  Left panel from left to right: state, partial derivative with respect to $x_1$, 
    and their linear combination.  
    Right panel: state over an annular domain.}
    \label{fig:intro2}
\end{figure}

\subsection{Relation to existing work}\label{sec-litrew}

Recent work on enforcing boundary constraints takes a spectral approach by relating the GRF covariance operator to the elliptic operators. \citet{Solin2019Know} first introduced a new covariance structure $k^{\partial \mathcal{D}}(\boldsymbol{x},\boldsymbol{x}')= \sum_j s\left(\lambda_j^{-1/2}\right)\psi_j(\boldsymbol{x})\psi_j(\boldsymbol{x}')$, where $s$ is the spectral density of a stationary, isotropic covariance function, and $\lambda_j$ and $\psi_j$ for $j=1,2,...$ are eigenvalues and eigenfunctions of the Laplace operator $\Delta$, respectively. For a domain $\mathcal{D}$ with boundary $\partial \mathcal{D}$, $\lambda_j$ and $\psi_j$ are obtained by solving the eigenvalue problem $-\Delta\psi_j(\boldsymbol{x})=\lambda_j\psi_j(\boldsymbol{x})$ for $\boldsymbol{x}\in \mathcal{D}$, and $\psi_j(\boldsymbol{x})=0$ for $\boldsymbol{x}\in \partial \mathcal{D}$. Since $\psi_j$ approach $0$ as $\boldsymbol{x}$ approaches $\partial \mathcal{D}$, the new class of covariance functions $k^{\partial \mathcal{D}}$ define GRFs with fixed-state boundary constraints. However, as pointed out by \citet{Solin2019Know} and \citet{Gulian2022BVP}, except for fixed-state boundary constraints on rectangular or spherical domains, or certain types of mixed boundary constraints on rectangular domains, the eigenvalues and eigenfunctions typically need to be approximated via discretization. Although appropriate in their specific context, the approximation means that boundary enforcement for the desired domains is no longer exact, resulting in less realistic modeling. Furthermore, the spectral approach is incompatible with nonstationary, anisotropic base covariance functions, and constraints on segments of the boundary. See \citet{Padilla2025PhysicsInformedGP} for another spectral construction that applies a numerical method for approximating constrained covariances.

Another existing approach is to directly derive analytical expressions for GRF covariance functions to attain fixed-endpoint constraints on a $1$-dimensional domain at one or both endpoints by repeatedly integrating an initial stationary covariance. Such expressions, derived from the squared exponential and uniform covariances, and the Mat\'ern covariance were obtained in \citet{Chkrebtii2013Thesis} and \citet{Ye2022MultiFidelity}, respectively. In addition to being restricted to cases where such analytical integrals are available, extensions to multi-dimensional domains are restricted to a product-form covariance specification, making the approach limited in practice. Adopting a different strategy, \citet{Ding2019BdryGPAN} link the Mat\'ern covariance with smoothness hyperparameter $0.5$ to a differential operator and derive a closed-form state-constrained covariance on rectangular domains, called \textit{BdryMat\'ern}. \citet{Ding2025BdryMaternGP} extend BdryMat\'ern to certain mixed boundary constraints on irregular (i.e, non-hypercube) domains through a stochastic partial differential equation formulation. This construction provides Mat\'ern-type smoothness control and approximation error analysis, but generally requires finite element and Monte Carlo approximations to evaluate the constrained covariance on non-rectangular domains.

Another direct approach to covariance specification, introduced in \cite{Tan2018BI}, formulated the covariance as $k^{\partial \mathcal{D}}(\boldsymbol{x},\boldsymbol{x}')=\phi(\boldsymbol{x})\phi(\boldsymbol{x}')\rho(\boldsymbol{x},\boldsymbol{x}')$, using a correlation function $\rho$, and an approximate distance function (ADF) $\phi$ that approaches $0$ as $\boldsymbol{x}$ approaches the boundary, to enforce fixed-state boundary constraints. \citet{Dalton2024BoundaryConstrainedGP} extended this approach and provided recipes for enforcing certain types of mixed boundary constraints by further restricting the ADF, referred to as normalization. They draw connections between directly specifying the covariances and transforming existing GRFs on a case-by-case basis. In practice, curved domain boundaries are approximated using a piecewise linear function, and each piece is assigned an individual ADF, all of which are joined into a global ADF. Besides introducing approximation, as noted by \cite{Dalton2024BoundaryConstrainedGP}, the normalization of the global ADF, a necessary condition for enforcing fixed-derivative constraints, fails at the segment joints. Furthermore, from a modeling perspective, it is important to directly address smoothness \citep{Ding2025BdryMaternGP}.

\citet{Gasbarra2007GB} generalize Brownian bridges to \textit{Gaussian bridges}, illustrated in the top left panel of Figure \ref{fig:condvsour}. 
These are Gaussian processes on $1$-dimensional intervals that are restricted at both endpoints through conditioning on the boundary values. \citet{Gasbarra2007GB} show that this conditioned measure is well defined. However, this approach does not straightforwardly generalize to higher-dimensional domains. Indeed, a related proposal by \citet{Vernon2019JUQ} for imposing continuous constraints across edges of a rectangular domain by conditioning on a restriction over an uncountable set lacks an analogous justification. Alternatively, \citet{Wang2016Shape} suggest conditioning a GRF on a large (but finite) set of synthetic error-free observations to enforce constraints across segments of spatial domains. However, although GRFs can be quite flexible in interpolating those observations, the required covariance matrix inversion quickly becomes computationally prohibitive. More importantly, this method cannot exactly enforce continuous, infinite-dimensional boundary constraints, as illustrated in the top row of Figure \ref{fig:condvsour}. 

In contrast to existing methods that either define constrained covariance structures in limited settings or rely on approximations of domains or constraints, our approach directly transforms an existing GRF into one that exactly satisfies the desired boundary constraints. While our framework requires parametrizable domain boundaries, we illustrate through one example that it is also compatible with parametric approximations of the boundary when this function is not available. Moreover, we explicitly account for smoothness and provide sufficient conditions on the transformation to enforce linear boundary constraints on multi-dimensional domains. 

\begin{figure}[t]
    \centering
    \fbox{\includegraphics[width=0.25\linewidth]{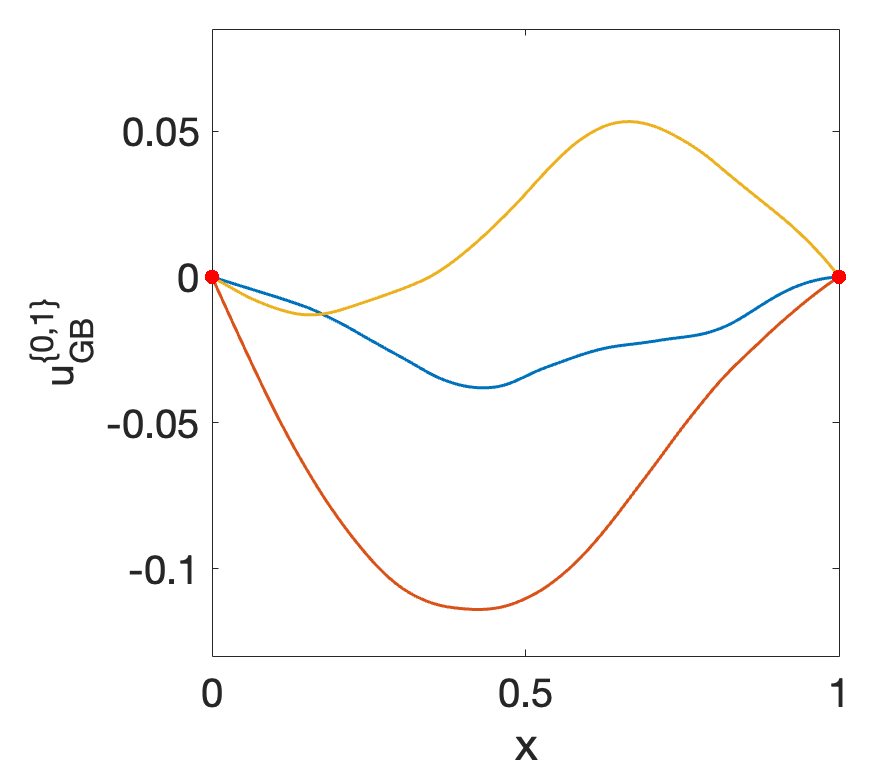}}
    \fbox{\includegraphics[width=0.68\linewidth]{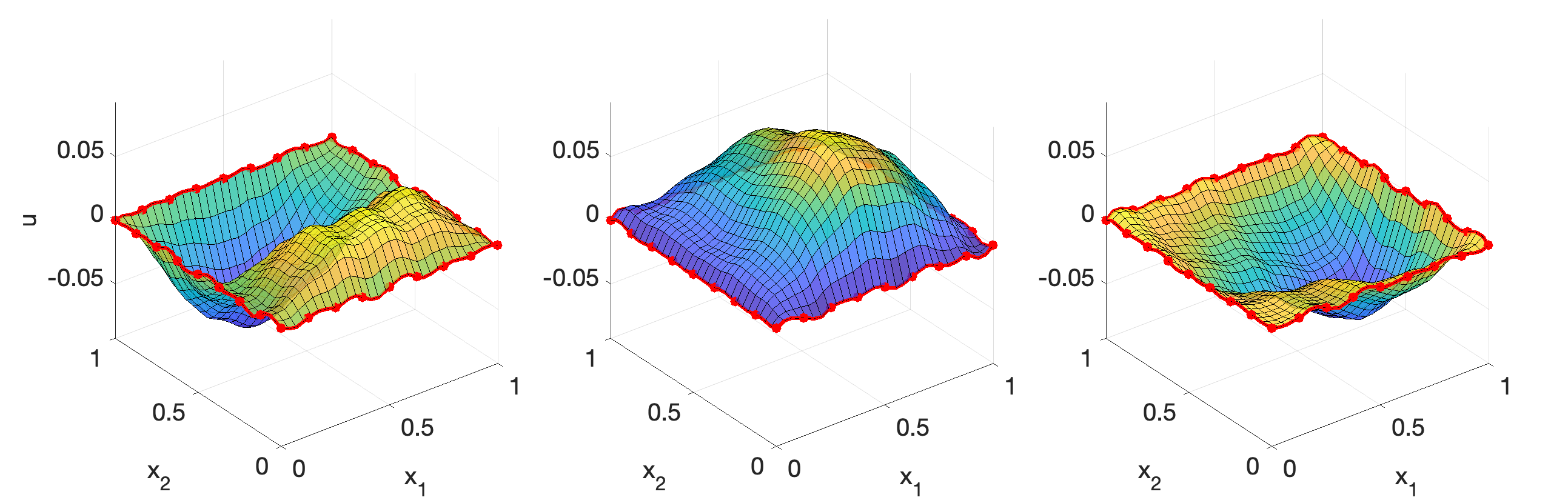}}
    \fbox{\includegraphics[width=0.25\linewidth]{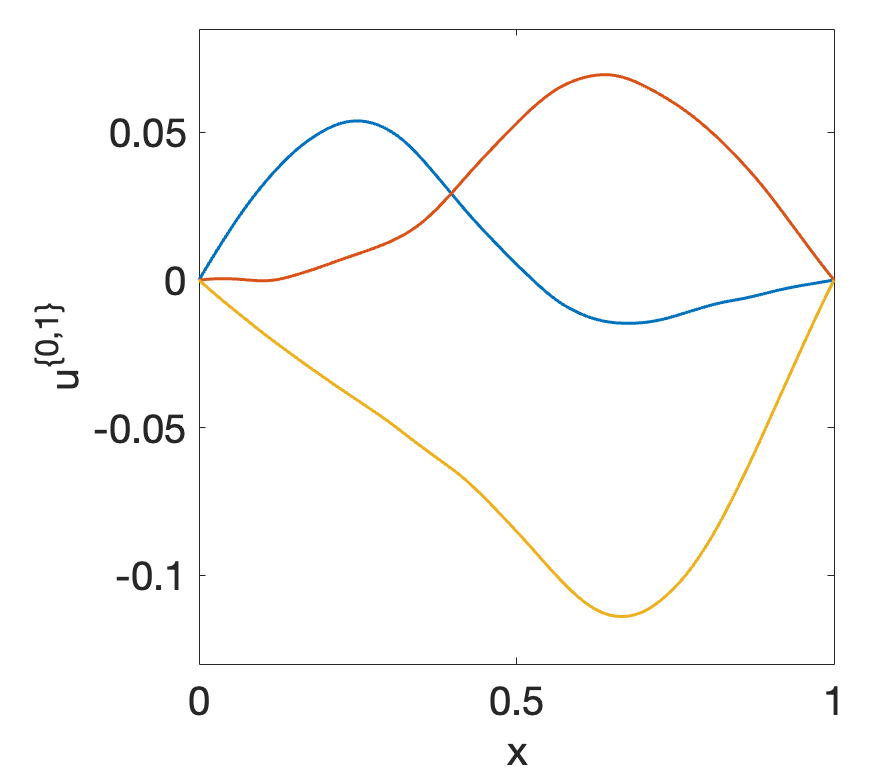}}
    \fbox{\includegraphics[width=0.68\linewidth]{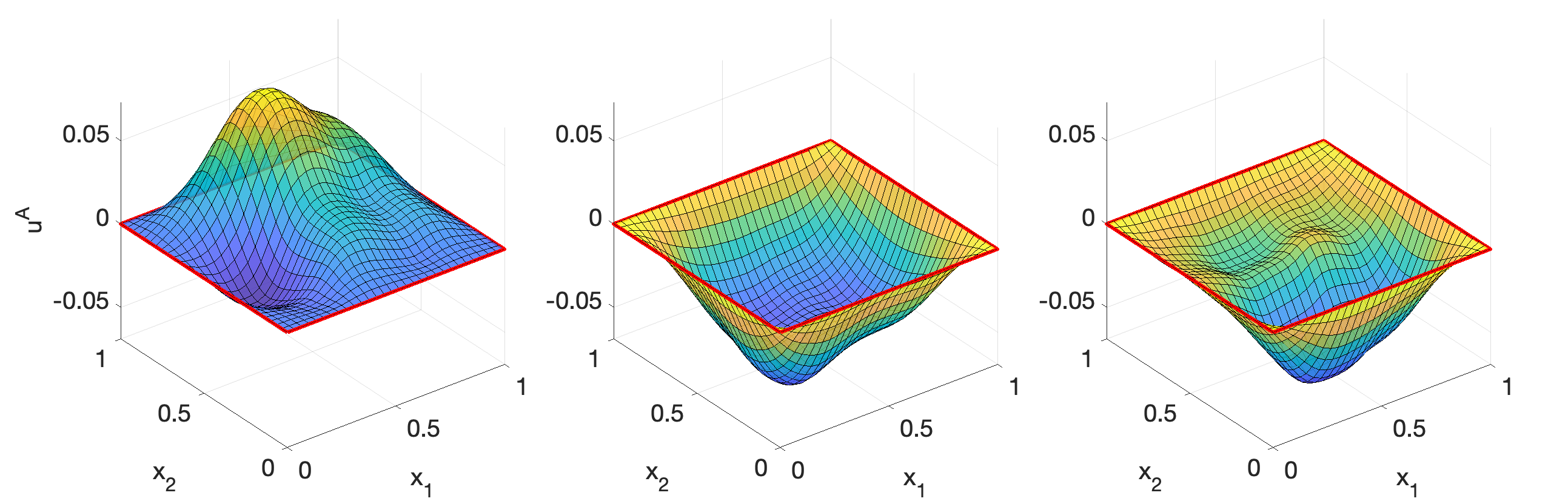}}
    \caption{Top left panel: three samples from a Gaussian bridge from $(0,0)$ to $(1,0)$; Top right panel: three samples from a point-wise conditioned GRF; Bottom panels: three samples from cGRFs on $1$- and $2$-dimensional spatial domains with continuously enforced boundary constraints. Red lines and points represent the state and synthetic observations at the boundary, respectively.}
    \label{fig:condvsour}
\end{figure}

\subsection{Contributions}

Our approach for constructing linearly boundary-constrained GRFs is fundamentally different from the existing methods summarized in Section \ref{sec-litrew}, although some of their resulting covariance structures naturally arise from our approach, as shown in the supplement.  The first contribution in Section \ref{sec:method} is a definition of a constrained random field satisfying continuous linear boundary restrictions, which we refer to interchangeably as a \emph{boundary-constrained} or \emph{linearly-constrained} random field. This definition leads to our main contribution, which is a representation-based framework to construct constrained GRFs by transforming suitably smooth unconstrained base GRFs. We call this class of models \textit{cGRF} and our approach the \textit{cGRF framework}. Our constructive approach for defining cGRFs retains some properties of the base GRF, such as smoothness, while enforcing constraints at the boundaries. Crucially, constructing cGRFs does not require approximation as long as a boundary parametrization is available and the corresponding mean and covariance functions can be obtained in closed form. In this setting, the cGRF framework is compatible with general linear boundary constraints on multi-dimensional domains, resulting in flexible functional priors for a variety of applications. As examples, we employ cGRFs as prior models of system states on problems in probabilistic numerics \citep{Cockayne2019ReviewProbN, Chkrebtii2016ProbSolver} and data-driven discovery of dynamical systems \citep{North2023ReviewDiscPDE, North2025DiscoveryPDE, Chen2020Lulu}. In these cases, enforcing known constraints on the system states leads to more realistic models, which translates into improved estimation accuracy. Lastly, we apply our approach to the problem of boundary-constrained inference on the displacement field over a material sample from discrete tensile test measurements, which is an important step for ensuring that materials used in industrial application meet safety and performance standards.

\section{Background}

We briefly review Gaussian random field (GRF) models and their use as priors over unknown states for Bayesian function estimation. This is followed by a review of a class of GRFs on a $1$-dimensional domain called  \emph{Gaussian bridges} \citep{Gasbarra2007GB}, which impose fixed-state constraints at the endpoints. We note that our proposed framework is both fundamentally different and more general than the Gaussian bridge.

GRFs are probability models of smooth states over a spatio-temporal domain $\mathcal{D}$, which are fully specified by their mean  $m(\boldsymbol{x}) = \mathbb{E}[u(\boldsymbol{x})]$ and covariance function $k(\boldsymbol{x}, \boldsymbol{x'}) = \mathbb{E}[(u(\boldsymbol{x})-m(\boldsymbol{x}))(u(\boldsymbol{x'})-m(\boldsymbol{x'}))]$ for $\boldsymbol{x},\boldsymbol{x}' \in \mathcal{D}$. We say that $\{u(\boldsymbol{x}),\boldsymbol{x}\in\mathcal{D}\}$ is a GRF, denoted by $u(\boldsymbol{x}) \sim GRF(m(\boldsymbol{x}), k(\boldsymbol{x}, \boldsymbol{x'}))$, if any finite collection $\boldsymbol{u} = (u(\boldsymbol{x}_1),...,u(\boldsymbol{x}_n))$ follows a multivariate normal distribution $\mathcal{N}(\boldsymbol{\mu},K)$ with the mean vector and covariance matrix having entries $\boldsymbol\mu_i=m(\boldsymbol{x}_i)$ and $K_{ij}=k(\boldsymbol{x}_i,\boldsymbol{x}_j)$, respectively, for $i,j=1,...,n$. 
The covariance function controls the degree of dependence in the state at nearby locations on the domain. An example of a flexible stationary covariance structure is the Mat\'ern function  
\begin{equation*}
    k(\boldsymbol{x}, \boldsymbol{x}')=\alpha^{-1}\frac{1}{\Gamma(\nu)2^{\nu-1}}\left(\frac{\sqrt{2\nu}\,||\boldsymbol{x}-\boldsymbol{x}'||}{\lambda}\right)^{\nu}K_{\nu}\left(\frac{\sqrt{2\nu}\,||\boldsymbol{x}-\boldsymbol{x}'||}{\lambda}\right), \quad \alpha, \lambda, \nu>0, 
\end{equation*}
where $K_{\nu}$ is the modified Bessel function of the second kind. A special case called the squared exponential covariance is obtained in the limit as $\nu \to \infty$. GRFs with squared exponential covariance are infinitely sample-path differentiable, and GRFs with Mat\'ern covariance are $\lceil \nu \rceil-1$ times sample-path differentiable \citep{Paciorek2003Nonstationary, Wang2021NonlinearPDE}. 
The hyperparameters $\alpha, \lambda$, and $\nu$ determine the prior variance, length-scale, and degree of smoothness, respectively.

An important property of GRFs which enables our framework is that they retain Gaussianity under linear operations $\mathcal{L}:\mathcal{H}\to\mathcal{H}'$, including integration, differentiation, and composition. Indeed, if $u$ is a GRF with mean  $m$ and covariance $k$ such that $\mathcal{L}u$ exists, then $\mathcal{L}u$ is also a GRF with mean $\mathcal{L}m$ and covariance $\mathcal{L}k\mathcal{L}^*$, where  $\mathcal{L}^*$ denotes the adjoint of $\mathcal{L}$. For example, if $\mathcal{L}$ is the differential operator, and $m$ and $k$ are such that $u$ is sample-path differentiable, then the GRF $\mathcal{L}u$ is well defined. As another example, the composition of a GRF with a deterministic function 
$f:\mathcal{D}\to\mathcal{D}$ 
remains a GRF.

Pointwise boundary enforcement is possible by conditioning a prior GRF measure on error-free measurements $\boldsymbol{y} = \mathcal{L}\boldsymbol{u}$ at finitely many boundary points $\boldsymbol{x}\in \partial\mathcal{D}$. If $u$ is a GRF prior with mean $m_0$ and covariance $k_0$, the posterior measure is likewise a GRF with mean and covariance,
\begin{equation}
    \label{naivecond}
    \begin{split}
    m(\boldsymbol{x}) &= m_0(\boldsymbol{x})+k_0\mathcal{L}^*(\boldsymbol{x},X)\,\left[\mathcal{L}  k_0\mathcal{L}^*(X,X)\right]^{-1}\left(\boldsymbol{y}-\mathcal{L}m_0(X)\right), \\
    k(\boldsymbol{x}, \boldsymbol{x}') &= k_0(\boldsymbol{x}, \boldsymbol{x}') - k_0\mathcal{L}^*(\boldsymbol{x},X)\,\left[\mathcal{L} k_0\mathcal{L}^*(X,X)\right]^{-1} \mathcal{L}  k_0(X, \boldsymbol{x}').
    \end{split}
\end{equation}
Here $k(X,X')$ for a covariance $k$ denotes an $n\times m$ matrix with $(i,j)$th element $k(\boldsymbol{x}_i,\boldsymbol{x}_j)$ where $X = (\boldsymbol{x}_1,...,\boldsymbol{x}_n)$ and $X' = (\boldsymbol{x}'_1,...,\boldsymbol{x}'_m)$. Similar updates are available for conditioning on measurements of the state across the domain with additive Gaussian error. These closed-form updates make GRFs natural candidates for active learning algorithms for sequential experimental design \citep{Santner2003CompExp, Chen2020Lulu} and smoothing problems in spatial statistics \citep{Stein1999Interpolation, Banerjee2004Spatial, Schabenberger2017Spatial}. 

The \textit{Gaussian bridge} \citep{Gasbarra2007GB} is one example on a $1$-dimensional domain, where a Gaussian process (GP) is conditioned on two synthetic measurements at the endpoints, as illustrated in the top left panel of Figure \ref{fig:condvsour}. The Gaussian bridge defines constrained GPs through a representation, from which corresponding mean and covariance functions may be computed (we will show that the Gaussian bridge representation arises as a special case of the cGRF for a 1-dimensional domain). Consider the interval $[0,T]$, two scalars $\xi$ and $\theta$, and a Gaussian process $u^{\{0\}}_{GB}$ with $u^{\{0\}}_{GB}(0)=\xi$. \citet{Gasbarra2007GB} provides the following proposition. Let $u^{\{0\}}_{GB}$ be a Gaussian process with mean function $m^{\{0\}}_{GB}$ and covariance function $k^{\{0\}}_{GB}$. Then the Gaussian bridge $u^{\{0,T\}}_{GB}$ from $(0,\xi)$ to $(T,\theta)$ admits a representation
$u^{\{0,T\}}_{GB}(x) = u^{\{0\}}_{GB}(x)+ k^{\{0\}}_{GB}(T,x)\, k^{\{0\}}_{GB}(T,T)^{-1} \left(\theta-u^{\{0\}}_{GB}(T)\right)$ 
with mean and covariance functions
\begin{align*}
    \begin{split}
        m^{\{0,T\}}_{GB}(x)&=m^{\{0\}}_{GB}(x)+k^{\{0\}}_{GB}(T,x)\,k^{\{0\}}_{GB}(T,T)^{-1}\left(\theta-m^{\{0\}}_{GB}(T)\right), \\
        k^{\{0,T\}}_{GB}(x,x') &= k^{\{0\}}_{GB}(x,x')-k^{\{0\}}_{GB}(T,x)\,k^{\{0\}}_{GB}(T,T)^{-1}k^{\{0\}}_{GB}(T,x').
    \end{split}
\end{align*}
Generalizations to higher-dimensional domains are not available, as they require proving that measures conditioned on states over a continuous boundary are well defined.
We instead adopt a constructive approach, based on linear transformation of a suitably smooth base GRF, which constrains the states without explicitly conditioning.

\section{Methodology}
\label{sec:method}

We begin by defining what we mean for a random field to be \textit{linearly-constrained} over a segment of the domain boundary. Throughout the presentation, we define random fields on a probability space $(\Omega, \mathcal{F}, P)$ over a spatial domain $\mathcal{D}\subset\mathbb{R}^d$ with boundary $\partial \mathcal{D}$, taking values in a Hilbert space $\mathcal{H}$. Moreover, we write that two random fields equal to each other when they equal in law. When we say that a random field is sample-path continuous or sample-path differentiable, we mean that there exists a modification that is sample-path continuous or sample-path differentiable \citep[see, for e.g.,][]{Seeger2004GPML, Paciorek2003Nonstationary, Wang2021NonlinearPDE}. 

\begin{definition}
    \label{def:constrGRF}
    Let $\mathcal{L}:\mathcal{H}\to \mathcal{H}'$ be a linear operator, $g:\mathcal{D}\to \mathbb{R}$ be a \textit{target} function, and $A\subset \partial\mathcal{D}$ be a segment of the boundary.  
    We say that a random field $u^A: \mathcal{D}\times \Omega \to \mathbb{R}$ is linearly-constrained if $\mathcal{L}u^A$ is sample-path continuous and $\mathcal{L}u^A(\boldsymbol{x})$ equals $g(\boldsymbol{x})$ almost surely for all $\boldsymbol{x}\in A$.
\end{definition}

While Definition \ref{def:constrGRF} includes any sample-path continuous random field that satisfies the boundary constraint with probability one, this work focuses on the Gaussian case in particular. We construct a linearly-constrained Gaussian random field, termed \emph{cGRF}, as
\begin{equation}
    \label{eqn:cGRF}
    u^A \sim GRF(m^A,k^A) := cGRF(\mathcal{L}, g, A, m_0, k_0),
\end{equation}
by transforming an unconstrained \textit{base GRF} $u\sim GRF(m_0,k_0)$ into a GRF with mean $m^A$ and covariance $k^A$. Note that in order to satisfy Definition \ref{def:constrGRF}, $m_0$ and $k_0$ must be chosen such that $\mathcal{L}u^A$ is sample-path continuous. We refer the reader to supplement Table 1 for examples of base covariance choices that are appropriate under common linear operators.

In the remainder of Section \ref{sec:method}, we develop representations of cGRF models of states under various linear boundary constraints on multi-dimensional domains. Section \ref{sec:3.1} introduces the \textit{projection} and \textit{weight} functions used to construct cGRFs. Section \ref{sec:3.2} presents our framework for constructing cGRF representations along with closed-form expressions for their  mean and covariance functions. Section \ref{sec:3.3} provides a recipe for constructing weight functions and examines the effect of weight function choice on the resulting cGRFs.

\subsection{Preliminaries for constructing cGRFs}
\label{sec:3.1}

In this section, we discuss different types of linear boundary constraints and introduce the projection and weight functions used in the representations in Section \ref{sec:3.2}. For simplicity in this section we take $g=0$,  and consider more general choices in Section \ref{sec:3.2}. To build an intuitive understanding of the approach, we begin by constructing a cGRF on a $1$-dimensional interval with a single endpoint constraint. Hereafter,  we will assume that the base GRF $u$ is suitably smooth, i.e. that the choice of $m_0$ and $k_0$ are such that $\mathcal{L}u$ is sample-path continuous.

\begin{example}
    \label{ex:u0}
    We wish to construct a GRF constrained to $g=0$ at the left endpoint $A=\{0\}$ on the $1$-dimensional domain $\mathcal{D}=[0,1]$, denoted by $u^{\{0\}}\sim cGRF(\mathcal{I},0,\{0\},m_0,k_0)$, where $\mathcal{I}$ is the identity operator, and $m_0$ and $k_0$ are the mean and covariance of a base GRF $u$. Consider the transformation $u^{\{0\}}(x)=u(x)-u(0)$ of $u$. It is clear that $u^{\{0\}}(x)=0$ almost surely at $x=0$, and $u^{\{0\}}$ is as smooth as $u$ across the domain. 
    Its mean and covariance functions are $m^{\{0\}}(x)=m_0(x)-m_0(0)$, and $k^{\{0\}}(x,x')=k_0(x,x')-k_0(x,0)-k_0(0,x')+k_0(0,0)$, respectively. The supplement shows that $k^{\{0\}}$ is the covariance derived directly in \citet{Chkrebtii2013Thesis} and \citet{Ye2022MultiFidelity}.
\end{example}

Section \ref{sec:3.2} describes a framework for enforcing more general constraints over multi-dimensional domains. For this, we next introduce the projection and weight functions $f$ and $w$.

\subsubsection{Projection functions for multi-dimensional domains}

Constraints on multi-dimensional domains are typically defined over continuous boundary segments, such as the sides of a rectangular domain or segments of the circular boundary of a disk-shaped domain. Recall that $A\subset \partial \mathcal{D}$ denotes the segment or set of segments where one wishes to enforce constraints. Below we use examples to introduce the projection function $f$ and relate it to $A$. Implementation details for cGRFs on the unit disk and unit triangle are provided in the the supplement.

\begin{example}
    \label{ex:Af}
    For the unit square domain $\mathcal{D}=[0,1]^2$, three examples of boundary segments $A$ are the left side $\{(x_1,x_2)\in \partial\mathcal{D} \mid x_1=0\}$, the two parallel sides $\{(x_1,x_2)\in \partial\mathcal{D} \mid x_1=0 \text{ or } x_1=1\}$, and the two adjacent sides $\{(x_1,x_2)\in \partial\mathcal{D} \mid x_1=0 \text{ or } x_2=0\}$, illustrated in red in the first three panels of Figure \ref{fig:Af}.
    For the unit disk domain $\mathcal{D}=\{(x_1,x_2)\in \mathbb{R}^2 \mid x_1^2+x_2^2\leq 1\}$, two examples of boundary segments $A$ are the half circle $\{(x_1,x_2)\in \partial\mathcal{D} \mid x_1 = -(1-x_2^2)^{1/2}\}$, and the full circle 
    $\{(x_1,x_2)\in \partial\mathcal{D}\}$,
    illustrated by the fifth and sixth panels of Figure \ref{fig:Af}. For the \textit{V-notched} square domain $\mathcal{D}=\{(x_1,x_2)\in[0,1]^2\mid x_1 \ge 0.5-|x_2-0.5|\}$, an example of boundary segment $A$ is the angled left side $\{(x_1,x_2)\in\partial\mathcal{D} \mid x_1 = 0.5 - |x_2-0.5|\}$, illustrated in the fourth panel of Figure \ref{fig:Af}. For the annular domain $\mathcal{D}=\{(x_1,x_2)\in\mathbb{R}^2\mid 0.2^2 \leq x_1^2+x_2^2 \leq 1\}$, an example of the boundary segment $A$ is two concentric circles $\{(x_1,x_2)\in \partial\mathcal{D}\}$, illustrated in the last panel of Figure \ref{fig:Af}.
\end{example}

For a given boundary segment $A$, the cGRF construction requires a continuous function $f:\mathcal{D}\to A\subset \mathcal{D}$ mapping interior points in $\mathcal{D}$ onto $A$ and mapping points in $A$ onto themselves, i.e., $f(\boldsymbol{x})=\boldsymbol{x}$ when $\boldsymbol{x}\in A$. Notice that $u^{\{0\}}(x)$ in Example \ref{ex:u0} can be equivalently represented as $u(x)-u\circ f(x)$ with $f(x)=0$, for all $x\in[0,1]$.

\begin{figure}[t]
    \centering
    \includegraphics[width=1\linewidth]{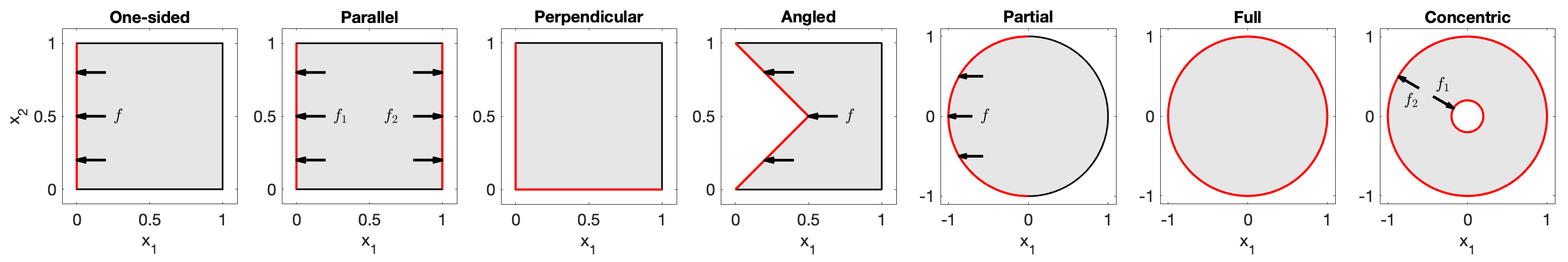}
    \caption{Seven examples of boundary segments $A$ (red lines) and the corresponding projection function in five cases (black arrows).}
    \label{fig:Af}
\end{figure}

\noindent\textbf{Example \ref{ex:Af} (Continued).}
\textit{For the unit square domain, we construct a cGRF constrained to $g=0$ at the left side $A=\{(x_1,x_2)\subset \partial\mathcal{D} \mid x_1=0\}$ via the representation $u^{A}(x_1,x_2)=u(x_1,x_2)-u\circ f(x_1,x_2)$ with $f(x_1,x_2) = (0,x_2)$, for all $(x_1,x_2) \in \mathcal{D}$. For the unit disk domain, a cGRF constrained to $g=0$ at the half circle $A=\{(x_1,x_2)\subset \partial\mathcal{D} \mid x_1 = -(1-x_2^2)^{1/2}\}$ is obtained as $u^{A}(x_1,x_2)=u(x_1,x_2)-u\circ f(x_1,x_2)$ with $f(x_1,x_2) = (-(1-x_2^2)^{1/2},x_2)$, for all $(x_1,x_2) \in \mathcal{D}$. For the V-notched square domain, a cGRF constrained to $g=0$ at the angled left side $\{(x_1,x_2)\in\partial\mathcal{D} \mid x_1 = 0.5 - |x_2-0.5|\}$ is obtained as $u^{A}(x_1,x_2)=u(x_1,x_2)-u\circ f(x_1,x_2)$ with $f(x_1,x_2) = (0.5 - |x_2-0.5|,x_2)$, for all $(x_1,x_2) \in \mathcal{D}$. In all three settings, $f$ is the projection of interior points $\boldsymbol{x}$ to the boundary segment $A$, along the direction $(-1,0)$ (first, fourth and fifth panels of Figure \ref{fig:Af}, respectively). 
}

\begin{example}
\label{ex:nonconvex}
    (Non-parametrized domains) Closed-form expressions for cGRF projection functions 
    require a parametrization of the domain boundary. When the boundary is defined by a set of points, we use parametric interpolation, as illustrated in Figure \ref{fig:nonconvex} on the outline of County Laois in Ireland, considered in \citet{Dalton2024BoundaryConstrainedGP}. An 11-segment piecewise linear approximation is used for illustration and a closed-form covariance is provided in the supplement. Extensions to polynomial interpolation and to a larger number of segments follow straightforwardly. It is also important to note that boundary irregularity can reduce the smoothness of the resulting cGRF relative to the base GRF. 
\end{example}

\begin{figure}[t]
    \centering
    \fbox{\includegraphics[width=0.253\linewidth]{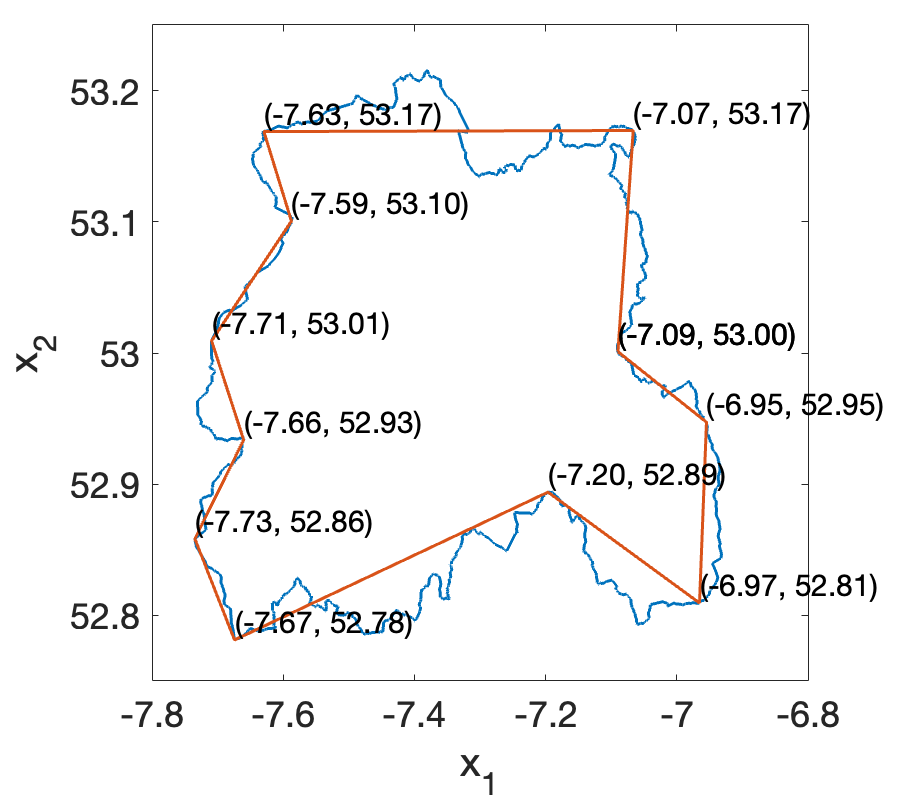}}
    \fbox{\includegraphics[width=0.708\linewidth]{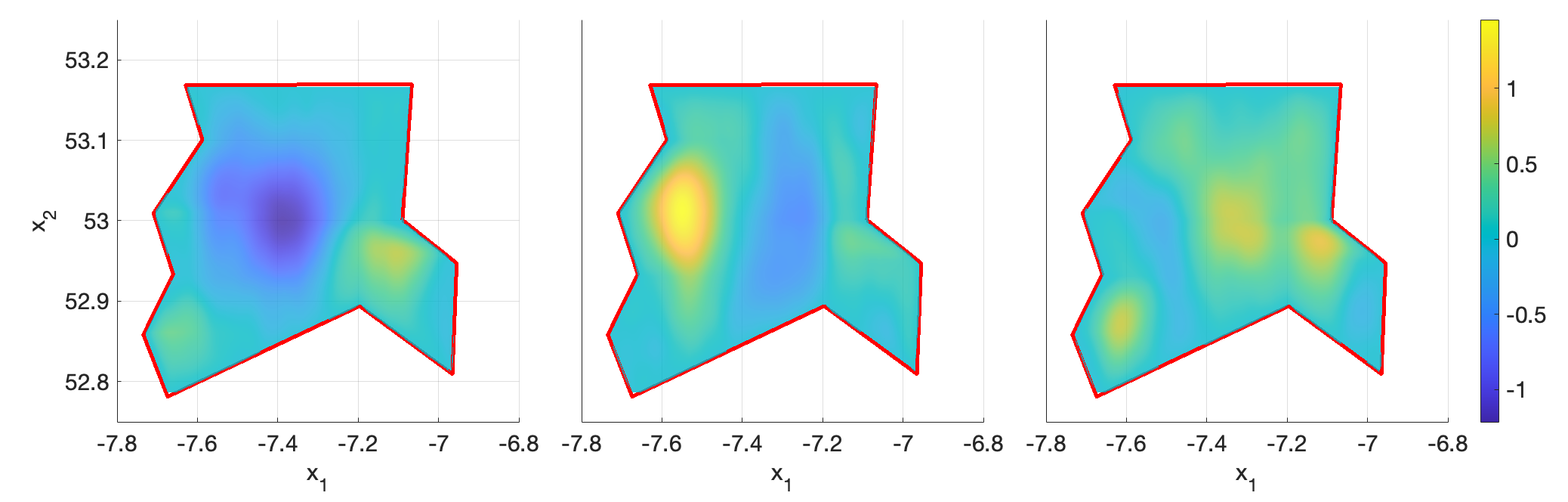}}
    \caption{
    Left panel:
     approximate parametrization using 11 segments (in orange) of a domain boundary defined by a set of points (in blue) of County Laois in Ireland. 
    Right panel: three samples from a cGRF with fixed-state constraints over the parametrized boundary.}
    \label{fig:nonconvex}
\end{figure}

\subsubsection{Weight functions for constraining linear transformations of the states}
\label{sec:3.1.2}

We introduce the weight function, $w: \mathcal{D} \to \mathbb{R}$, to achieve two other generalizations. The first is to impose a constraint involving a linear transformation $\mathcal{L}$ of the state. In this case, $w$ activates the constraint only when $\mathcal{L}$ is applied. 

\noindent\textbf{Example \ref{ex:u0} (Continued).}
\textit{Consider the first-order differential operator $\mathcal{L}=\partial_x$. Returning to the $1$-dimensional domain, we wish to construct $u^{\{0\}}\sim cGRF(\partial_x,0,\{0\}, m_0,k_0)$ to satisfy the constraint $\partial_xu^{\{0\}}(x)=0$ at $x=0$. One such representation is $u^{\{0\}}(x)=u(x)-w(x)(\partial_x u \circ f(x))$, where $f(x)=0$, and $w(x)=x$ for all $x\in [0,1]$. Notice that $\partial_x u^{\{0\}}(x) = \partial_x u(x)-\partial_x u (0)$, so that $\partial_x u^{\{0\}}$ is constrained to $g=0$ at $A=\{0\}$. The center panel in Figure \ref{fig:1dmixed} shows a cGRF with fixed-derivative boundary constraint at the right endpoint. To simultaneously enforce the constraint at the left endpoint, as shown in the right panel, refer to the discussion in Section \ref{sec:3.3}. The supplement shows that $u^{\{0\}}$ is related to Eq. $(27)$ in \citet{Dalton2024BoundaryConstrainedGP} for a fixed-derivative boundary constraint.}

\begin{figure}[t]
    \centering
    \includegraphics[width=1\linewidth]{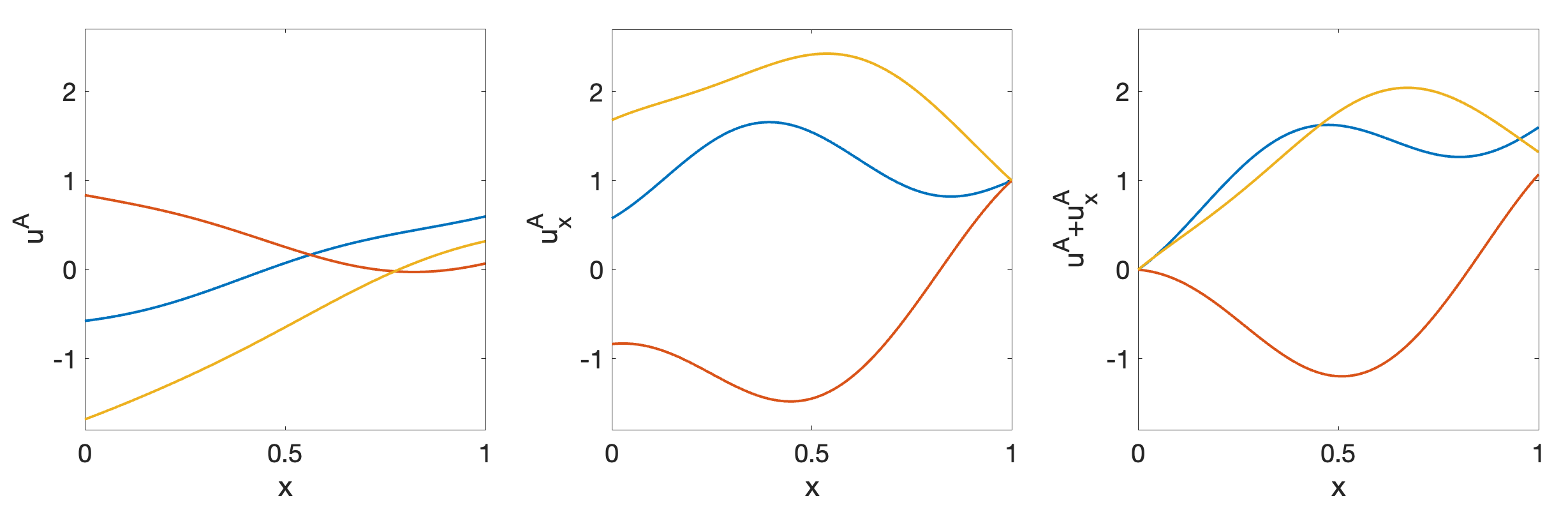}
    \caption{Marginal cGRF samples (blue, yellow, and orange lines) over state $u^A$, first derivative $u^A_x$, and $u^A+u^A_x$ (from left to right). The first derivative is constrained to $1$ at the right endpoint; while $u^A+u^A_x$ is constrained to $0$ at the left endpoint.}
    \label{fig:1dmixed}
\end{figure}

Another generalization consists of imposing different boundary constraints across disjoint boundary segments. Consider a set of $n$ weight functions, $\boldsymbol{w}=\{w_i\mid i=1,...,n\}$, where $n$ is the number of constraints. $\boldsymbol{w}$ activates the corresponding constraints as $\boldsymbol{x}$ approaches different segments of the boundary. See Section \ref{sec:3.3} for a recipe for constructing $\boldsymbol{w}$. 

This generalization is also convenient for relaxing the restriction on the projection function $f$ when a single projection for the entire set of boundary segments $A$ is unavailable, by treating the constraint on $A$ as multiple constraints of the same type defined on disjoint segments of $A$.

\noindent\textbf{Example \ref{ex:Af} (Continued).}
\textit{
For the two-sided parallel boundary constraint on the unit square domain, the function $f$ cannot simultaneously map $\boldsymbol{x}\in \mathcal{D}$ to the parallel sides $A$, while ensuring that $f(\boldsymbol{x})=\boldsymbol{x}$ at $A$. Hence, we partition $A$ into $A_{1}=\{(x_1,x_2)\in \partial \mathcal{D}\mid x_1=0\}$ and $A_{2}= \{(x_1,x_2)\in \partial \mathcal{D}\mid x_1=1\}$, and construct $u^{\{A_1,A_2\}}(x_1,x_2)=u(x_1,x_2)-w_{1}(x_1,x_2)(u\circ f_{1}(x_1,x_2)) - w_{2}(x_1,x_2)(u\circ f_{2}(x_1,x_2))$, where $f_{1}(x_1,x_2)=(0,x_2)$, $f_{2}(x_1,x_2)=(1,x_2)$,  $w_{1}(x_1,x_2)=1-x_1$, and $w_{2}(x_1,x_2)=x_1$, for $(x_1,x_2) \in \mathcal{D}$ (see the second plot of Figure \ref{fig:Af} for illustration of $f_1,f_2$). Notice that $\{w_{1},w_{2}\}$ approaches $\{0,1\}$ and $\{1,0\}$, as $x_1$ approaches $0$ and $1$, respectively, and that as a result, $u^{\{A_1,A_2\}}$ is constrained to $0$ across $A$. Similarly, a cGRF $u^{\{A_1,A_2\}}$ can be constructed on the annular domain by partitioning $A$ into inner and outer circles (last plot of Figure \ref{fig:Af}) using the radial projections $f_1(\boldsymbol{x})=0.2 \boldsymbol{x}/||\boldsymbol{x}||$ and $f_2(\boldsymbol{x})= \boldsymbol{x}/||\boldsymbol{x}||$, respectively, and the weight functions $w_{1}$ and $w_{2}$ constructed using the recipe provided in Section \ref{sec:3.3}.
}

\subsection{A general cGRF representation}
\label{sec:3.2}

This section formally introduces the cGRF framework. Theorem \ref{thm:constrGRF} provides sufficient conditions on the base GRF $u$, and the custom projection and weight functions introduced in Section \ref{sec:3.1}, for the constrained GRF $u^{\boldsymbol{A}}$ to satisfy $\mathcal{L}_iu^{\boldsymbol{A}}(\boldsymbol{x})=g_i(\boldsymbol{x})$ for all $\boldsymbol{x}\in A_i$, $i=1,...,n$. Explicit expressions for the mean and covariance functions of $u^{\boldsymbol{A}}$ are also provided.

\begin{theorem}
    \label{thm:constrGRF}
    For $i=1,...,n$, let $\mathcal{L}_i:\mathcal{H}\to \mathcal{H}_i' \subset \mathcal{C}(\mathcal{D})$ be a linear operator, let $g_i: \mathcal{D} \to \mathbb{R}$ be a target function, let $A_i \subset \partial\mathcal{D}$ be a boundary segment such that there exists a continuous projection function $f_i:\mathcal{D}\to A_i$ satisfying $f_i(\boldsymbol{x})=\boldsymbol{x}$ for all $\boldsymbol{x}\in A_i$, and let $w_i: \mathcal{D} \to \mathbb{R}$ be a weight function. 
    Let $u:\mathcal{D}\times \Omega \to \mathbb{R}$ be a Gaussian random field such that $u(\cdot ,\omega)\in \mathcal{H}$ for $P$-almost all $\omega\in\Omega$. 
    Furthermore, assume that for $i,j=1,...,n$, 
    \begin{equation}
        \label{eqn:condsmooth}
        w_j(\cdot)\left(g_j\circ f_j(\cdot) - \mathcal{L}_ju\circ f_j(\cdot, \omega)\right) \in \mathcal{H}
    \end{equation}
    for $P$-almost all $\omega\in\Omega$, and for all $\boldsymbol{x}\in A_i$, 
    \begin{equation}
        \label{eqn:cond}
        \mathcal{L}_i \left[w_j(\boldsymbol{x})\left(g_j\circ f_j(\boldsymbol{x})-\mathcal{L}_ju\circ f_j(\boldsymbol{x})\right)\right] = 
        \begin{cases}
            g_j\circ f_j(\boldsymbol{x})-\mathcal{L}_iu\circ f_j(\boldsymbol{x}), & j=i \\
            0, & j\neq i
        \end{cases} 
    \end{equation} 
    then a linearly-constrained Gaussian random field $u^{\boldsymbol{A}}$ can be constructed as
    \begin{equation}
        \label{eqn:rep}
        u^{\boldsymbol{A}}(\boldsymbol{x})=u(\boldsymbol{x})+ \sum_{j=1}^n w_j(\boldsymbol{x})\left(g_j\circ f_j(\boldsymbol{x})-\mathcal{L}_ju\circ f_j(\boldsymbol{x})\right).
    \end{equation} 
\end{theorem}

\begin{proof}
  The proof is provided in the supplement.
\end{proof}

We denote the resulting linearly-constrained cGRF model of the state satisfying the set of given linear boundary constraints by 
$u^{\boldsymbol{A}} \sim cGRF(\boldsymbol{\mathcal{L}},\boldsymbol{g},\boldsymbol{A}, m_0, k_0)$, where $\boldsymbol{\mathcal{L}}=\{\mathcal{L}_i\mid i=1,...,n\}$, $\boldsymbol{g}=\{g_i\mid i=1,...,n\}$, and $\boldsymbol{A}=\{A_i\mid i=1,...,n\}$. As before, $m_0$ and $k_0$ are the mean and covariance of the base GRF $u$. Note that the representation \eqref{eqn:rep} for a given set of linear boundary constraints is not unique since it is based on the choice of $\boldsymbol{f}$ and $\boldsymbol{w}$. Examples in Section \ref{sec:3.1} are special cases of \eqref{eqn:rep}. For convenience, $\boldsymbol{f}$ is often chosen to be Cartesian or radial projections. Section \ref{sec:3.3} provides a recipe for constructing $\boldsymbol{w}$ and examines the effect of several different choices.

An important advantage of Theorem \ref{thm:constrGRF} is that it does not require establishing the existence of probability measures conditioned on restrictions over a continuous set, such as constraints over multi-dimensional domain boundaries \citep[see, e.g.,][for a discussion of challenges in proving the existence of such conditioned measures]{Stuart2010InvProblem, Cockayne2019ReviewProbN}. 
Instead, our constructive approach uses well-known properties of linear transformations of GRFs to transform a base GRF into one that satisfies the desired constraints. Our approach also assumes appropriate choices of $k_0, m_0, w_j, g_j$ and $f_j$ such that $u(\cdot, \omega)\in \mathcal{H}$ and $w_j(\cdot)\left(g_j\circ f_j(\cdot) - \mathcal{L}_ju\circ f_j(\cdot, \omega)\right) \in \mathcal{H}$ for $j=1,...,n$. For example, Mat\'ern covariances are convenient candidates for $k_0$ when $\mathcal{H}\subset \mathcal{C}^{\boldsymbol{\beta}}(\mathcal{D})$, for $\boldsymbol{\beta} \in \mathbb{N}^d_0$ and $\boldsymbol{\mathcal{L}}$ involves linear partial differential operators. In addition, we also need to ensure that $m_0, w_j, g_j, f_j\in \mathcal{H}\subset\mathcal{C}^{\boldsymbol{\beta}}(\mathcal{D})$ for $j=1,...,n$ \citep[see, e.g.,][]{Wang2021NonlinearPDE}. Candidates for $\mathcal{H}$ and $\mathcal{H}_i'$ are Sobolev spaces with $p=2$. Candidates for $f_j$ are continuous functions projecting interior points to the boundary and boundary points to themselves, along Cartesian or radial directions. If necessary, we further partition $\boldsymbol{A}$ to accommodate such choices of $f_j$, as in Example \ref{ex:Af}.  
Lemma \ref{lemma} provides the mean and covariance of the cGRF $u^{\boldsymbol{A}}$ constructed in Theorem \ref{thm:constrGRF}.

\begin{lemma}
    \label{lemma}
    The mean and covariance function of $u^{\boldsymbol{A}} \sim cGRF(\boldsymbol{\mathcal{L}}, \boldsymbol{g}, \boldsymbol{A}, m_0, k_0)$ are 
    \begin{align*}
        m^{\boldsymbol{A}}(\boldsymbol{x}) &= m_0(\boldsymbol{x}) + \sum_{j=1}^n w_j(\boldsymbol{x}) \ (g_i(f_i(\boldsymbol{x})) - \mathcal{L}_jm_0(f_j(\boldsymbol{x}))),  \\
        k^{\boldsymbol{A}}(\boldsymbol{x}, \boldsymbol{x}') &= k_0(\boldsymbol{x}, \boldsymbol{x}') - \sum_{j'=1}^n w_{j'}(\boldsymbol{x}') \ k_0\mathcal{L}_{j'}^*(\boldsymbol{x}, f_{j'}(\boldsymbol{x}')) -  \sum_{j=1}^n w_j(\boldsymbol{x}) \ \mathcal{L}_jk_0(f_j(\boldsymbol{x}), \boldsymbol{x}') \\
        &+ \sum_{j=1}^n\sum_{j'=1}^n w_j(\boldsymbol{x}) \ w_{j'}(\boldsymbol{x}') \ \mathcal{L}_jk_0\mathcal{L}_{j'}^*(f_j(\boldsymbol{x}), f_{j'}(\boldsymbol{x}')).
    \end{align*}
\end{lemma}
\begin{proof}
    The proof follows directly from Theorem \ref{thm:constrGRF} and the properties of GRFs.
\end{proof}

An important consequence is that the cGRF mean and covariance can be computed in closed form without resorting to approximation. Thus, conjugate Gaussian inference with a cGRF prior has the same scaling as standard Gaussian process regression, namely growing as the cube of the number of observations at most.

Figures \ref{fig:intro2}, \ref{fig:1dmixed} and \ref{fig:tri_periodic} illustrate samples from various cGRFs. Another interesting consequence of Theorem \ref{thm:constrGRF} is illustrated in the right panel of Figure \ref{fig:tri_periodic}. For a rectangular domain, the constrained covariance $k^{A}(\boldsymbol{x},\boldsymbol{x}')=k_1(x_1,x_1')  k_2^{A}(x_2,x_2')$ inherits the product covariance structure from the base covariance $k_0(\boldsymbol{x},\boldsymbol{x}')=k_1(x_1,x_1')  k_2(x_2,x_2')$, as shown in the supplement for the case where $k_1$ and $k_2$ are periodic and Mat\'ern covariances. As a result, the two samples from the resulting cGRF shown in the right panel of Figure \ref{fig:tri_periodic} satisfy periodic and fixed-state constraints in dimensions $x_1$ and $x_2$, respectively. In general, however, the cGRF framework accommodates any form of multi-dimensional base covariance $k_0$ under the conditions in Theorem \ref{thm:constrGRF}. 

\begin{figure}[t]
    \centering
    \fbox{\includegraphics[width=0.48\linewidth]{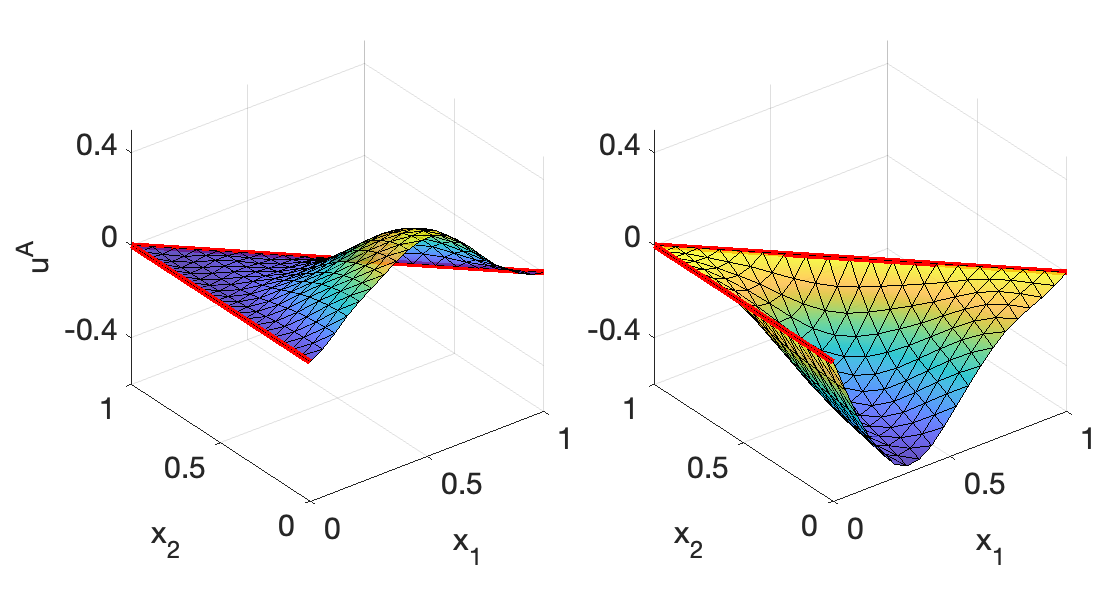}}
    \fbox{\includegraphics[width=0.48\linewidth]{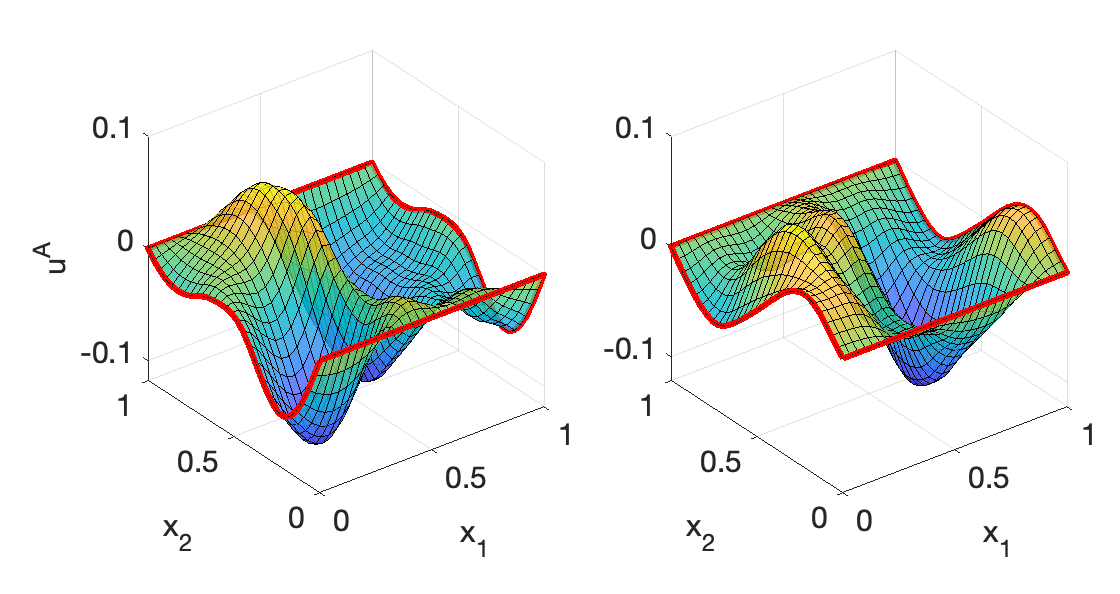}}
    \caption{Two samples from a cGRF on the triangular domain (left panel) and the unit square domain (right panel; base covariances have product structure with periodic and Mat\'ern components, respectively). Red solid lines indicate boundary constraints.}
    \label{fig:tri_periodic}
\end{figure}

\subsection{Constructing weight functions}
\label{sec:3.3}

Next, we provide some guidance on how to select the weight functions $\boldsymbol{w}$ in practice. The use of it is illustrated in the following examples, as well as for constructing fixed-state cGRFs over triangular or disk-shaped domains as described in the supplement.

\begin{conjecture}
\label{conjecture1}
Consider the set of linear boundary constraints  $\mathcal{L}_iu^{\boldsymbol{A}}(\boldsymbol{x})=g_i(\boldsymbol{x})$ over $\boldsymbol{x}\in A_i$ for $i=1,...,n$. 
For a given choice of projection functions 
$f_1,\ldots, f_n$
corresponding to boundary segments $A_1,\ldots, A_n$, consider the vector of weight functions
\begin{equation}
\label{eqn:w}
    \boldsymbol{w}(\boldsymbol{x})=\boldsymbol{v}^\top(\boldsymbol{x})M^{-1}(\boldsymbol{x}),
\end{equation}
where $\boldsymbol{v}:\mathcal{D}\to\mathbb{R}^{n}$ and $M:\mathcal{D}\to \mathbb{R}^{n \times n}$ are functions with elements $\boldsymbol{v}_i(\boldsymbol{x})=k_0\mathcal{L}_{i}^*(\boldsymbol{x}, f_{i}(\boldsymbol{x}))$ and $M_{ij}(\boldsymbol{x})=\mathcal{L}_{i}k_0\mathcal{L}_{j}^*(f_{i}(\boldsymbol{x}),f_{j}(\boldsymbol{x}))$, respectively, for $i,j=1,...,n$.
Then, the cGRF representation \eqref{eqn:rep} from Theorem \ref{thm:constrGRF} enforces the desired boundary constraints as long as 
conditions \eqref{eqn:condsmooth} and \eqref{eqn:cond} are satisfied. 
\end{conjecture}

As an example, when $\mathcal{D}=[0,T]$, $\boldsymbol{\mathcal{L}}=\{\mathcal{I},\mathcal{I}\}$, $\boldsymbol{g}=\{\xi,\theta\}$, and $\boldsymbol{A}=\{0,T\}$, the cGRF representation based on this choice of $\boldsymbol{w}$ is equivalent to the Gaussian bridge representation.

Notice that as compared to the choices of $\boldsymbol{w}$ in Section \ref{sec:3.1.2}, the recipe we propose here supports a more systematic way of constructing different types of cGRFs, at the expense of being more complicated. The optimal structure of $\boldsymbol{w}$ should be decided case by case. The following examples enforce mixed boundary constraints based on this choice of $\boldsymbol{w}$.

\begin{example}
\label{ex:w}
Consider the unit square domain $\mathcal{D}=[0,1]^2$, and the mixed boundary constraints $\mathcal{L}_iu^{\{A_1,A_2\}}(\boldsymbol{x})=0$ over $\boldsymbol{x} \in A_{i}$, for $i=1,2$, where $\mathcal{L}_i = a_i\partial_{x_1} + b_i$ with $a_i,b_i\in \mathbb{R}$, $A_{1}=\{\boldsymbol{x} \in \partial \mathcal{D} \mid x_1=0 \}$, and $A_{2}=\{\boldsymbol{x} \in \partial \mathcal{D} \mid x_1=1 \}$. Let $f_1(\boldsymbol{x})=(0,x_2)$, and $f_2(\boldsymbol{x})=(1,x_2)$. A cGRF satisfying the mixed boundary constraints is obtained using \eqref{eqn:w} with $\boldsymbol{w}(\boldsymbol{x})=\boldsymbol{v}^\top(\boldsymbol{x})M^{-1}(\boldsymbol{x})$, where $\boldsymbol{v}_i(\boldsymbol{x})=k_0\mathcal{L}_{i}^*(\boldsymbol{x}, f_{i}(\boldsymbol{x}))=a_i \, \partial_{x_1}k_0(\boldsymbol{x},f_i(\boldsymbol{x})) + b_i \, k_0(\boldsymbol{x},f_i(\boldsymbol{x}))$, and $M_{ij}(\boldsymbol{x}) = \mathcal{L}_ik_0\mathcal{L}_j^*(f_i(\boldsymbol{x}), f_j(\boldsymbol{x})) = a_ia_j \, \partial_{x_1}\partial_{x_1'}k_0(f_i(\boldsymbol{x}), f_j(\boldsymbol{x}))+ a_ib_j \, \partial_{x_1} k_0(f_i(\boldsymbol{x}), f_j(\boldsymbol{x})) + a_jb_i \, \partial_{x_1'}k_0(f_i(\boldsymbol{x}), f_j(\boldsymbol{x})) +b_ib_j \, k_0(f_i(\boldsymbol{x}), f_j(\boldsymbol{x}))$, for $i,j=1,2$. Here $\partial_{x_1}$ and $\partial_{x_1'}$ differentiate $k_0$ with respect to the first and second arguments, respectively.  As special cases, when $a_0=a_1=0$, the cGRF enforces fixed-state constraints. When $b_0=b_1=0$, the cGRF enforces fixed-derivative constraints. The mixed boundary constraints here also include the one-sided constraints, when $a_0=b_0=0$ or $a_1=b_1=0$. See Figure \ref{fig:1dmixed} and the left panel in Figure \ref{fig:intro2} for cGRFs with mixed boundary constraints on $1$- and $2$-dimensional domains.
\end{example}

\begin{example}
    \label{ex:variant}
    Certain boundary constraints that fall outside the scope of Theorem \ref{thm:constrGRF} 
    by violating condition \eqref{eqn:condsmooth} can be re-formulated to fit the cGRF framework. 
    Suppose we want to enforce $\partial_{x_1}u^A(x_1,x_2)=0$ at $A=\{(x_1,x_2)\subset \partial\mathcal{D} \mid x_2=0\}$, on $\mathcal{D}=[0,1]^2$. 
    Let $f(x_1,x_2)=(x_1,0)$. Notice that $\partial_{x_1}u \circ f(x_1,x_2)=\partial_{x_1}u(x_1,0)$ is less smooth than $u(x_1,x_2)$ in the $x_1$ dimension, 
    violating condition \eqref{eqn:condsmooth}. Nevertheless, we can enforce this boundary constraint indirectly by translating it into $u(x_1,x_2)=c$ at $A$ for some constant $c$, if known.
    Condition \eqref{eqn:condsmooth} is satisfied for the corresponding cGRF $u^A(x_1,x_2)=u(x_1,x_2)+w(x_1,x_2) \ (c-u\circ f(x_1,x_2))$, where $w(x_1,x_2) = k((x_1,x_2),(x_1,0)) \ k((x_1,0),(x_1,0))^{-1}$ is obtained using equation \eqref{eqn:w}. The application on tensile testing in Section \ref{sec-app} demonstrates the enforcement of this type of boundary constraint.
\end{example}

Lastly, we conduct a sensitivity analysis of the cGRF with respect to the choice of weight function parametrization to illustrate the effect of different specifications.

\begin{example}
    \label{ex:sensitivity}
    Consider five GRFs constrained to $g=0$ over the outer boundary of an annular domain with base GRF having mean zero and squared exponential covariance with $\lambda=0.5$ and $\alpha=1$. The radial projection function to the outer boundary is given by  $f(\boldsymbol{x})=\boldsymbol{x}/||\boldsymbol{x}||$ for $\boldsymbol{x}\in\mathcal{D}$. The weight functions are $w(\boldsymbol{x})=\exp(-||\boldsymbol{x}- f(\boldsymbol{x})||^2/2h^2)$, with $h=0.1,0.2,0.5,1,2$, respectively, each corresponding to a column in Figure \ref{fig:nonconvex1}. Note also that when $h=\lambda=0.5$, the weight function corresponds to the recipe in equation \eqref{eqn:w}. All weight functions equal one on the boundary but differ in their rate of decay from the interior. Smaller values of $h$ produce faster decay, so the boundary constraint is localized near the boundary and the interior quickly resembles the unconstrained GRF. Larger values of $h$ produce slower decay, so the boundary constraint influences the behavior of the state over a larger portion of the domain. In other words, the choice of weight function is not only intrinsic to exact constraint enforcement, but also a modeling choice that interacts with other covariance hyperparameters. In particular, the weight decay parameter $h$ controls how far the boundary correction propagates into the interior, while the covariance length-scale $\lambda$ controls how strongly interior points are correlated with the boundary points. Moreover, the weight function multiplicatively determines the magnitude of the boundary correction term, which may inflate interior variance and produce the flower-shaped sample structure reflecting the radial projection. This interaction is further complicated by the regularity and smoothness of the domain boundary, as illustrated in Figure \ref{fig:nonconvex}. Thus, the practical behavior of a cGRF depends jointly on the base covariance, the weight function, and the regularity of the domain boundary (assuming convenient choices of projection functions such as Cartesian or radial projections), as suggested by condition \eqref{eqn:condsmooth} in Theorem \ref{thm:constrGRF}.
\end{example}

\begin{figure}[t]
    \centering
    \fbox{\includegraphics[width=1\linewidth]{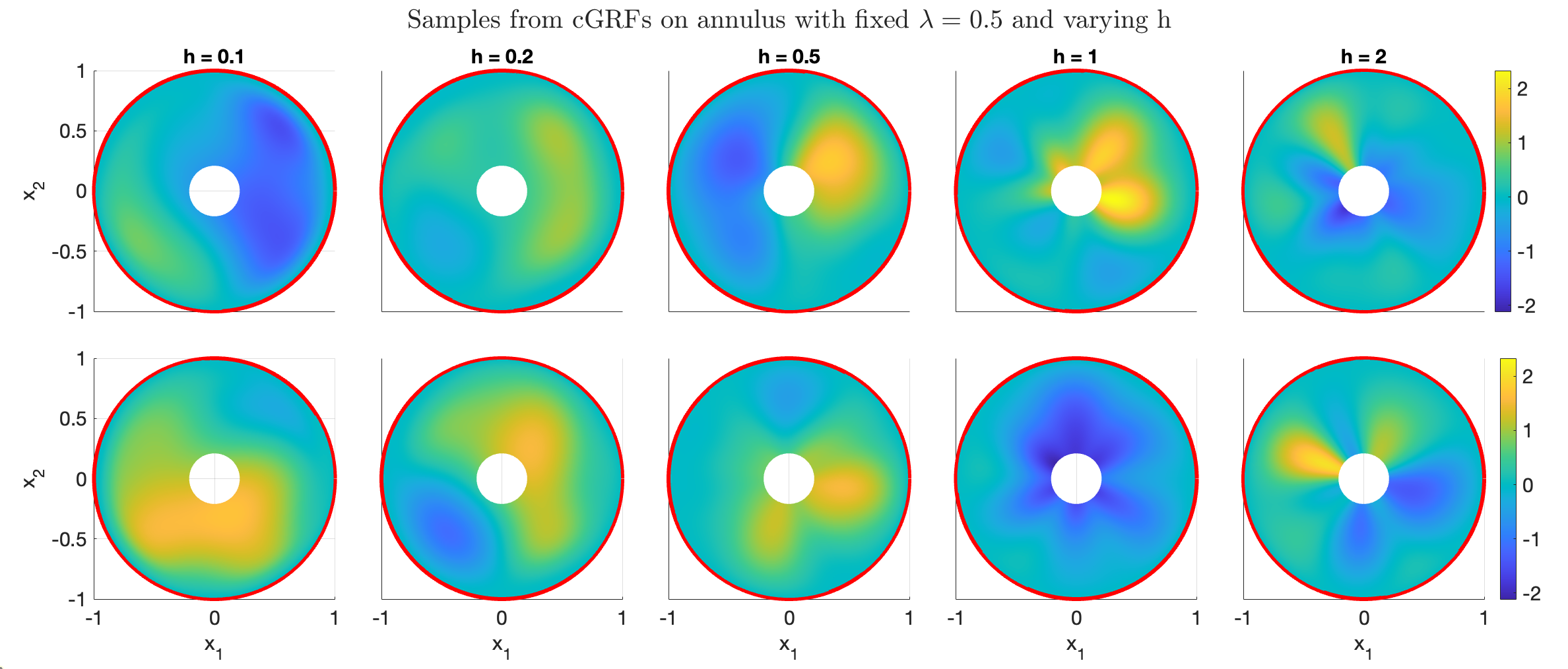}}
    \fbox{\includegraphics[width=1\linewidth]{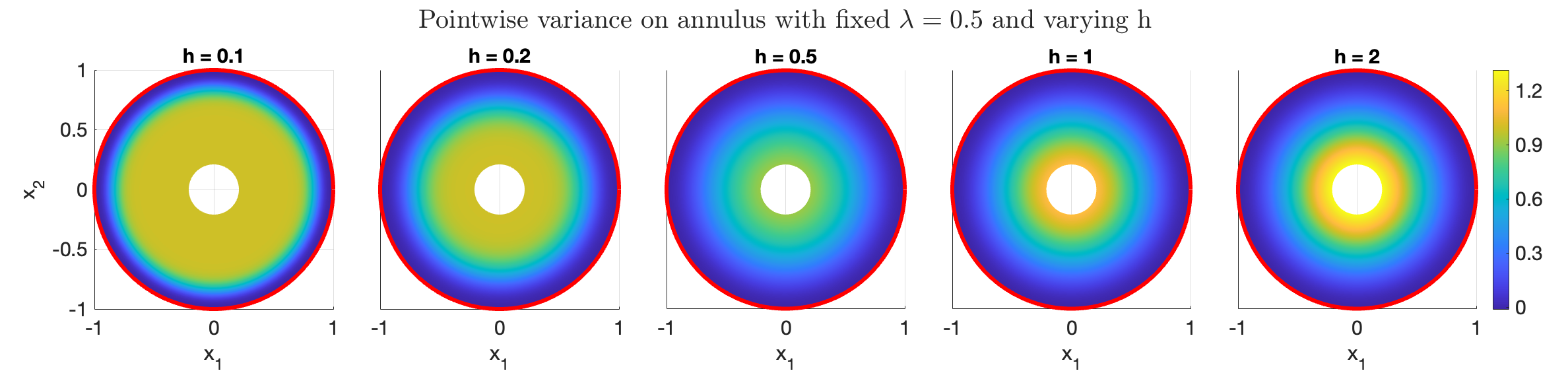}}
    \caption{
    Top panel: two samples from cGRFs with fixed-state constraints enforced on the outer circle. The weight functions $w(\boldsymbol{x})$ are $\exp(-||\boldsymbol{x}- f(\boldsymbol{x})||^2/2h^2)$ for $h = 0.1, 0.2, 0.5, 1, 2$, respectively, from left to right. The projection function $f(\boldsymbol{x})$ is the radial projection to the outer circle.
    Bottom panel: pointwise variance for each cGRFs.}
    \label{fig:nonconvex1}
\end{figure}

\section{Applications}\label{sec-app}

We consider some illustrative applications where replacing an unconstrained with a constrained probability model improves predictive performance and leads to more physically realistic results. Here we focus on exact GRF constraints via cGRF, and note that the results apply to exact approaches that can be expressed as special cases of the cGRF discussed in the supplement. We discuss the application of cGRFs to probabilistic numerical methods, data-driven discovery of dynamic systems, and inference for a displacement field from limited experimental measurements to study a material’s response to pulling forces. Since the temporal and spatial dimensions are often treated differently in applications, we will use the notation $u^{\boldsymbol{A}}(t,\boldsymbol{x})$ for cGRF models defined over a spatio-temporal domain.

\subsection{Probabilistic solvers for PDE boundary value problems}

Partial differential equations (PDEs) model states implicitly by relating them to their derivatives with respect to spatial or temporal variables. Inference requires an explicit representation of the state, or solution, which in many cases is only available numerically. The resulting unmodeled discretization uncertainty can lead to biased inference on calibration parameters, as well as over-confident uncertainty estimates. This has spurred recent interest in the field of probabilistic numerics (see \citet{Cockayne2019ReviewProbN} for a comprehensive review) which seeks to model this type of uncertainty via so-called \emph{probabilistic solvers} and propagate it through the statistical inverse problem. Careful modeling ensures that the resulting structured uncertainty agrees with the underlying dynamics. This requires choosing a prior over the state that reflects its known structure, including all boundary conditions. However, such priors were previously unavailable, and pointwise-conditioned GRFs were used, leading to the potential for numerical instability and less realistic modeling. 

\begin{figure}[t]
    \centering
    \includegraphics[width=1\linewidth]{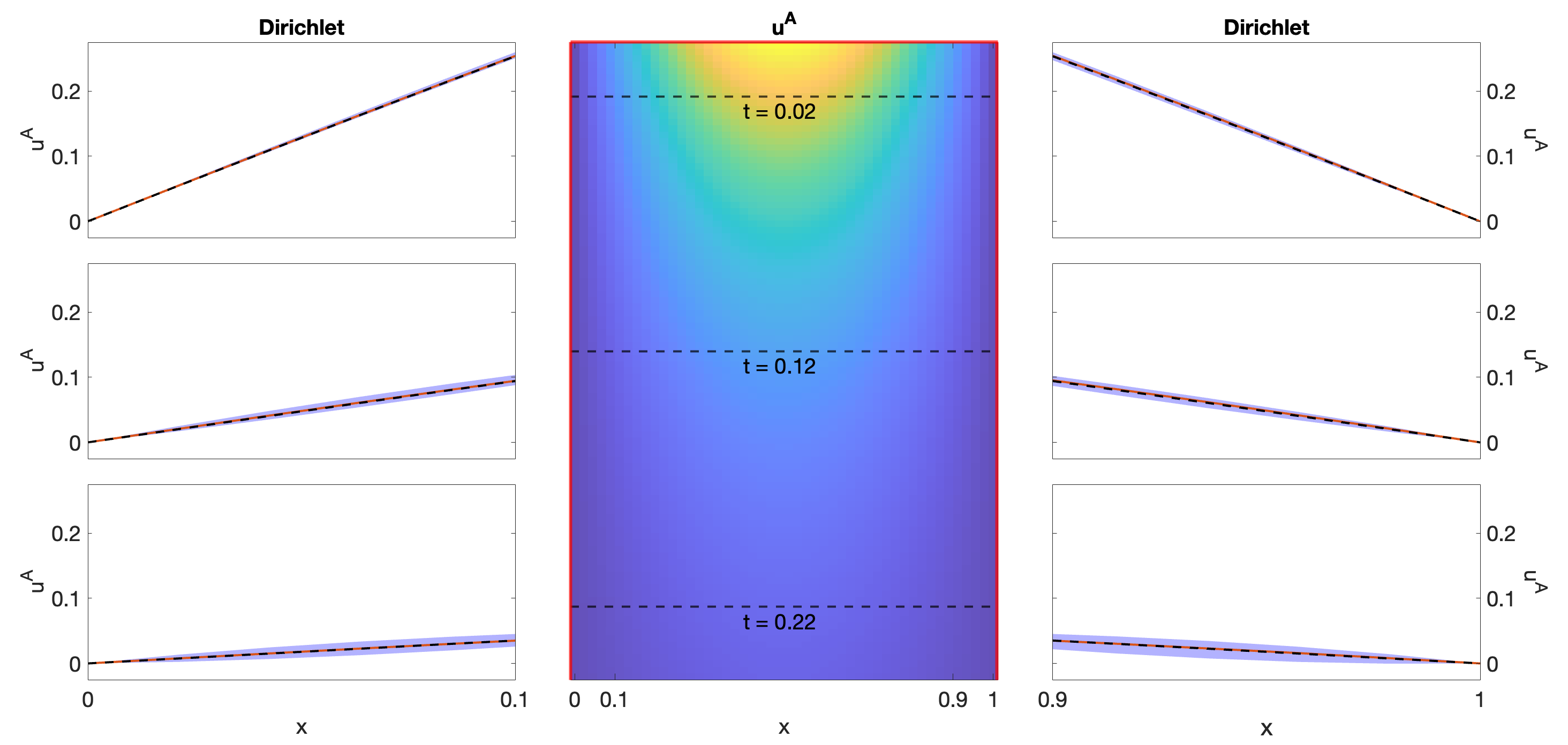}
    \caption{Center panel: mean probabilistic solution with initial condition $u(t,x) = \sin(\pi x)$ and Dirichlet boundary conditions along red solid lines. Left and right panels: empirical estimates of posterior uncertainty near $x=0$ and $x=1$ for three time steps marked by dashed lines in the center panel. Posterior means are shown in orange, pointwise $95\%$ credible intervals are represented by blue bands.}
    \label{fig:solverD}
\end{figure}

In one example, \citet{Chkrebtii2016ProbSolver} propose a state-space probabilistic solver which conditions a prior GRF model of the state on iterative evaluations of the PDE information operator. To illustrate the use of cGRFs in this setting, we consider the heat equation under two distinct sets of initial and boundary conditions over the temporal domain $[0,0.25]$ and the spatial domain $[0,1]$. Let $\mathcal{D} = [0,0.25]\times [0,1]$, and let $u:\mathcal{D}\to \mathbb{R}$ be the solution satisfying the PDE $u_{t}(t,x) = u_{xx}(t,x), \ (t,x)\in \mathcal{D}$, with boundary constraints
\begin{align}
\begin{split}
\label{f}
\mathcal{L}_1u(t,x) &= g_1(x), \ (t,x)\in A_1, \\
\mathcal{L}_2u(t,x) &= g_2(t), \ (t,x)\in A_2, \\
\mathcal{L}_3u(t,x) &= g_3(t), \ (t,x)\in A_3,
\end{split}
\end{align}
for linear operators $\boldsymbol{\mathcal{L}}=\{\mathcal{L}_1=\mathcal{I},\mathcal{L}_2,\mathcal{L}_3\}$ and target functions $\boldsymbol{g}=\{g_1,g_2,g_3\}$, to be defined next for two distinct examples. The boundary segments are $\boldsymbol{A}=\{A_1,A_2,A_3\}=\{\{0\}\times[0,1], (0,0.25]\times\{0\}, (0,0.25]\times\{1\}\}$. Following \citet{Chkrebtii2016ProbSolver}, we discretize the spatio-temporal domain by a $100\times16$ uniform grid, and carry out a ``forward in time, continuous in space'' sampling strategy which sequentially updates the joint prior $(u,u_t,u_{xx})\sim cGRF(\boldsymbol{\mathcal{L}},\boldsymbol{g}, \boldsymbol{A},m_0,k_0)$ on synthetic data about $u_{xx}$ across the domain. We use a squared exponential base covariance $k_0$ with length-scale hyperparameters $0.0175$ in the temporal dimension and $0.4573$ in the spatial dimension, proportional to the temporal and spatial step sizes, respectively. The prior precision hyperparameter is set to $1$. The probabilistic solution is obtained by marginalizing over the synthetic data and concentrates on the exact solution $u$ as the grid size grows. 

\begin{figure}[t]
    \centering
    \includegraphics[width=1\linewidth]{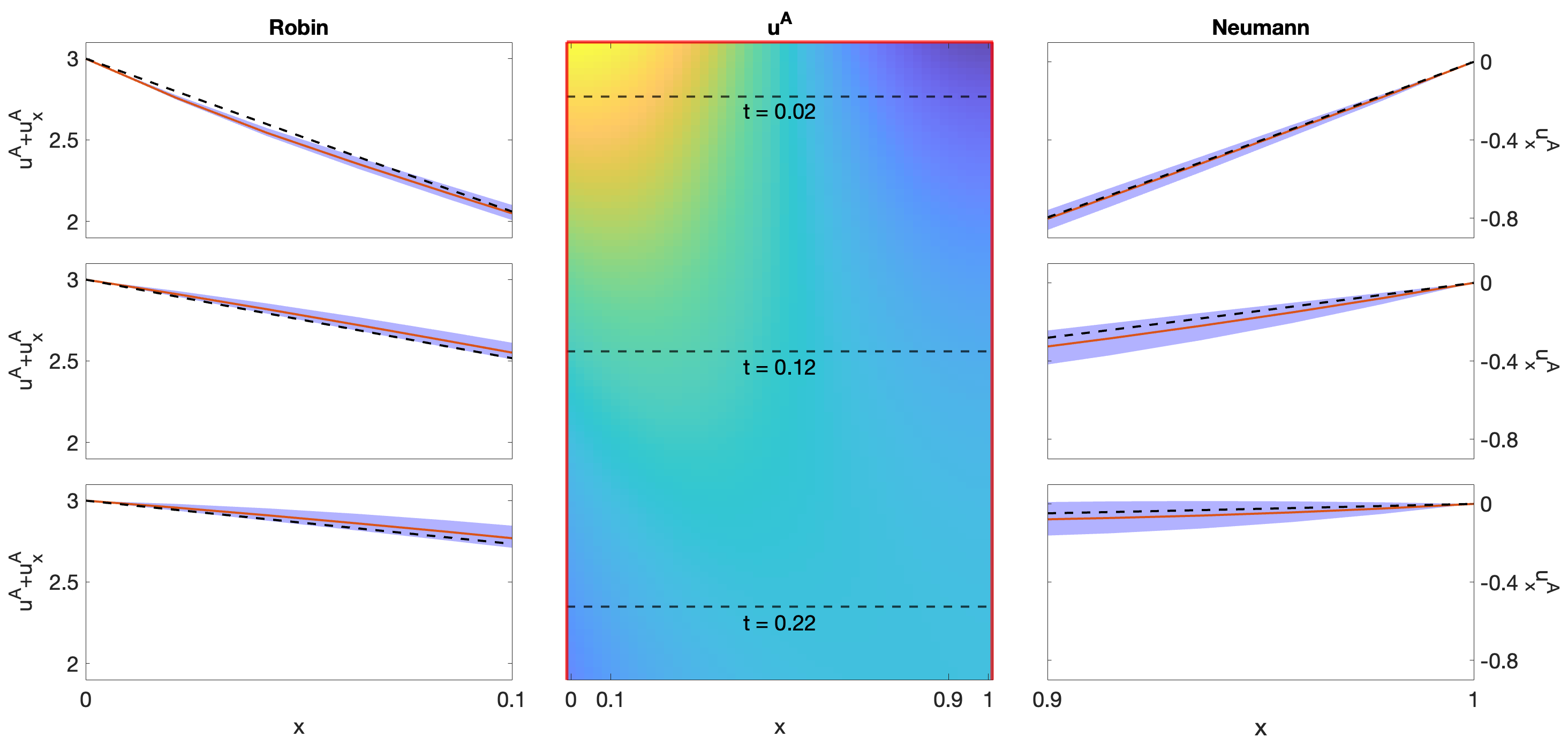}
    \caption{Center panel: mean probabilistic solution with initial condition $u(t,x) = \cos(\pi x) + 2$ and mixed Robin and Neumann boundary conditions along red solid lines. Left and right panels: empirical estimates of posterior uncertainty near $x=0$ and $x=1$ for three time steps marked by dashed lines in the center panel. Posterior means are shown in orange, pointwise $95\%$ credible intervals are represented by blue bands.}
    \label{fig:solverR}
\end{figure}

Figures \ref{fig:solverD} and \ref{fig:solverR} visualize discretization uncertainty for the two examples relative to the solution approximated via a high-order adaptive numerical method (MATLAB \texttt{pdepe} solver; black dashed line) for two examples with different boundary constraints. The first example uses the initial condition $u(t,x) = \sin(\pi x)$ and Dirichlet boundary conditions, $g_2(t)=g_3(t)=0$, and $\mathcal{L}_2=\mathcal{L}_3=\mathcal{I}$. The second example uses the initial condition $u(t,x) = \cos(\pi x) + 2$ and the mixed Robin and Neumann boundary conditions, $g_2(t)=3$, $g_3(t)=0$, and $\mathcal{L}_2=\mathcal{I}+\partial_x$, $\mathcal{L}_3=\partial_x$. The accumulation of discretization uncertainty is reflected in the widening of the credible intervals over time, while tending to zero as we approach the boundaries, reflecting the known physical mechanisms operating there.

\subsection{Data-driven discovery of dynamical systems}

When a reliable mathematical model for the state is unavailable, a data-informed approach may be used for selecting an appropriate dynamical system from a set of candidates to help practitioners understand the physical phenomenon and predict its future behavior \citep[e.g.,][]{North2023ReviewDiscPDE, North2025DiscoveryPDE}. For example, \citet{Chen2020Lulu} develop a GRF-assisted active learning algorithm that iteratively performs penalized linear regression of the temporal derivative on a collection of candidate model terms consisting of functions of the state and its spatial derivatives. These terms are evaluated at the input locations selected by an active learning algorithm. A GRF model of the latent state and its derivatives is used in the objective function to quantify the information gain from collecting data at potential input locations for choosing model components from a large set of candidates. Because GRFs do not in general enforce known boundary constraints, \citet{Chen2020Lulu} restricts attention to problems on unbounded domains. Introducing the known boundary constraints via cGRF modeling allows us to consider more realistic settings and can improve algorithm performance. 

\begin{figure}[t]
    \centering
    \includegraphics[width=1\linewidth]{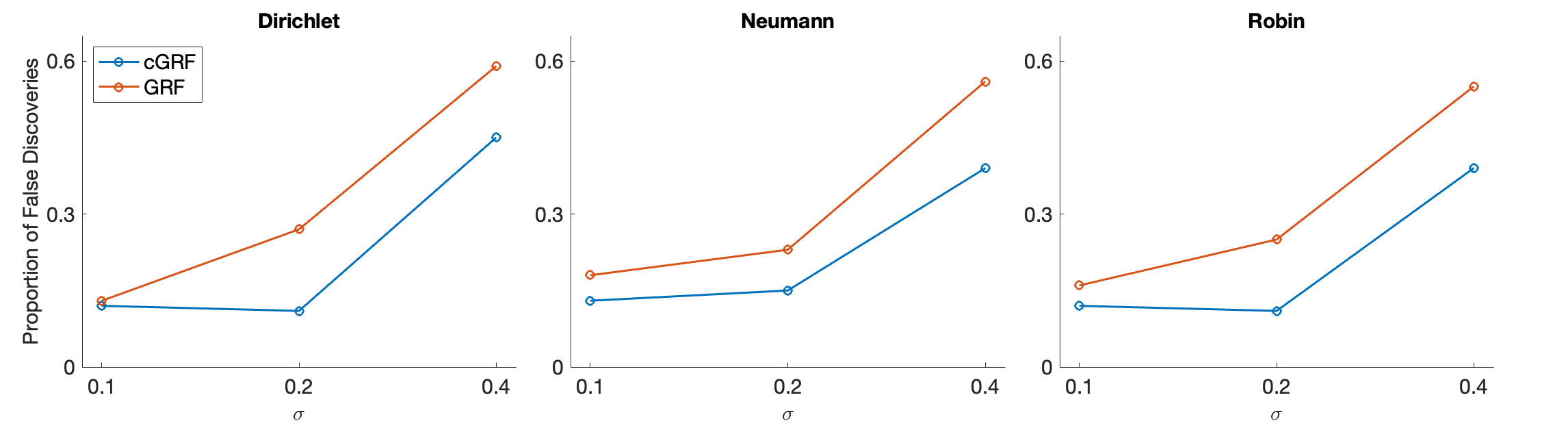}
    \caption{Proportion of false discoveries with cGRF (blue) and GRF (orange) at three levels of measurement error based on $100$ simulations under different boundary constraints (left to right).}
    \label{fig:percent_FD}
\end{figure}

\begin{figure}[ht]
    \centering
    \includegraphics[width=1\linewidth]{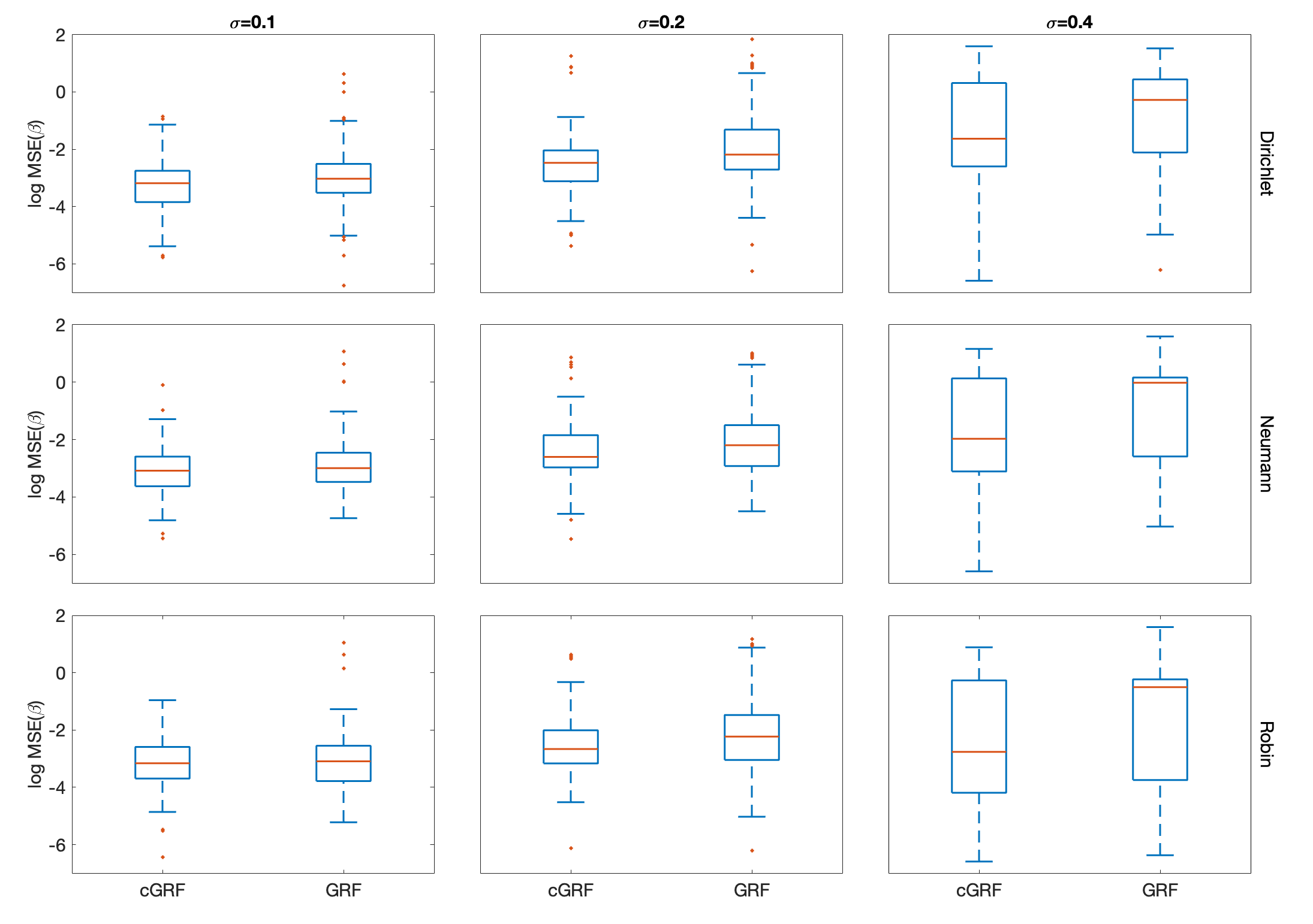}
    \caption{Log MSE for estimated coefficients based on $100$ simulations using $cGRF$ and $GRF$ under different settings. Measurement standard deviations are $\sigma=0.1,0.2$ and $0.4$ from left to right. Boundary condition types are two-sided Dirichlet, Neumann and Robin from top to bottom.}
    \label{fig:msebeta}
\end{figure}

As an illustration, we apply this approach to learn Burger's equation from simulated data under both a boundary-constrained cGRF model for the state and an unconstrained alternative. Consider the temporal domain $[0,1]$ and the spatial domain $[-10,10]$. Let $\mathcal{D}=[0,1]\times[-10,10]$, and let $u:\mathcal{D}\to \mathbb{R}$ be the underlying state satisfying the differential equation model $u_{t}(t,x) = u_{xx}(t,x)-uu_x(t,x), \ (t,x)\in \mathcal{D}$ with boundary constraints in \eqref{f}, for linear operators $\boldsymbol{\mathcal{L}}=\{\mathcal{L}_1=\mathcal{I},\mathcal{L}_2,\mathcal{L}_3\}$ and target functions $\boldsymbol{g}=\{g_1,g_2,g_3\}$, to be defined later for three distinct examples. The boundary segments are $\boldsymbol{A}=\{A_1,A_2,A_3\}=\{\{0\}\times[-10,10], (0,1]\times\{-10\}, (0,1]\times\{10\}\}$. We use tri-modal initial condition as in \citet{Chen2020Lulu} corresponding to $g_1(x)=2\exp(-15(x-9)^2)+1.5\exp(-15(x+1)^2)+\exp(-25(x+9)^2)$, $g_2(t)=\mathcal{L}_2g_1(-10)$, and $g_3(t)=\mathcal{L}_3g_1(10)$. Algorithm performance is compared when using the unconstrained $GRF(0,k_0)$ model and the constrained version,  $cGRF(\boldsymbol{\mathcal{L}},\boldsymbol{g},\boldsymbol{A},0,k_0)$. We test each setting under Dirichlet, Neumann, and Robin boundary conditions, corresponding to $\mathcal{L}_2=\mathcal{L}_3=\mathcal{I}$, $\mathcal{L}_2=\mathcal{L}_3=\partial_x$, and $\mathcal{L}_2=\mathcal{L}_3=\mathcal{I}+\partial_x$, respectively. We adopt the algorithm settings in \citet{Chen2020Lulu}, including using the squared exponential covariance $k_0$, the same collection of candidate terms $\{1, u, u^2, u^3, u_x,...,u_xu^2_{xx}, uu_xu_{xx}\}$, and levels of measurement error variance $\sigma^2=0.2^2,0.4^2$. The two evaluation criteria are the proportion of false discoveries including missing the two correct terms and selecting the remaining false terms, and the log mean squared error (MSE) of the estimated coefficients. As shown in Figures \ref{fig:percent_FD} and \ref{fig:msebeta}, performance in both metrics is improved when using cGRFs by accounting for the additional information available at the boundary. The supplement also provides a frequency plot for falsely identified terms based on $100$ simulations.

\subsection{Boundary-constrained state estimation}
\label{sec:4.3}

Materials used in industrial applications are subjected to a variety of tensile forces. The tensile test is a fundamental experimental method in materials science engineering for measuring a material's response to these forces. During the test, discrete measurements of the displacement field are recorded as a dog-bone-shaped specimen is stretched upward by a moving crosshead at a constant speed until fracture (see illustration in Figure \ref{fig:intro1}). The strain field is then calculated as the spatial derivative of the displacement field along the vertical ($x_2$) direction, and together with the stress measurements, is used to construct the stress–strain curve. This curve establishes critical material properties, including Young’s modulus, yield strength, ultimate tensile strength, and elongation at fracture, and its reliability depends on the accuracy of the displacement predictions at unobserved locations.
In this section, we evaluate the predictive performance of GRF-based estimation of displacement from discrete measurements with and without prior boundary constraint enforcement. The data, shown in Figure \ref{fig:intro1}, is simulated by adding independent normal errors to the displacement function obtained using the COMSOL software for multiphysics simulation.

Consider the temporal domain $[0,1]$ and the spatial domain $\mathcal{D}_s$, which is a dog-bone shaped spatial domain with lower boundary $A_1=\{(x_1,x_2) \mid x_2=-10\}$ and upper boundary $A_2=\{(x_1,x_2) \mid x_2=10\}$, marked in red in the left panel of Figure \ref{fig:intro1}. Denote $\mathcal{D}=[0,1]\times \mathcal{D}_s$, and $\boldsymbol{A}=\{A_1,A_2\}$. Let $u: \mathcal{D}\to \mathbb{R}$ be the displacement field in the vertical ($x_2$) direction. Then the boundary constraints are $u(t,x_1,x_2)=0.005t$ at $x_2=10$ and $u(t,x_1,x_2)=0$ at $x_2=-10$, meaning the top of the specimen at $A_2$ is pulled upwards by the moving crosshead at the constant speed $0.005$ mm/s, while the bottom side is held in place and remains stationary. Notice that no constraints are necessary on the displacement field along the vertical sides of the domain (segments not highlighted in red in the first panel of Figure \ref{fig:intro1}). Notice also that although the specimen is stretched, the spatial domain remains unchanged, because displacement is measured relative to the starting locations.

\begin{figure}[t]
    \centering
    \includegraphics[width=0.821\linewidth]{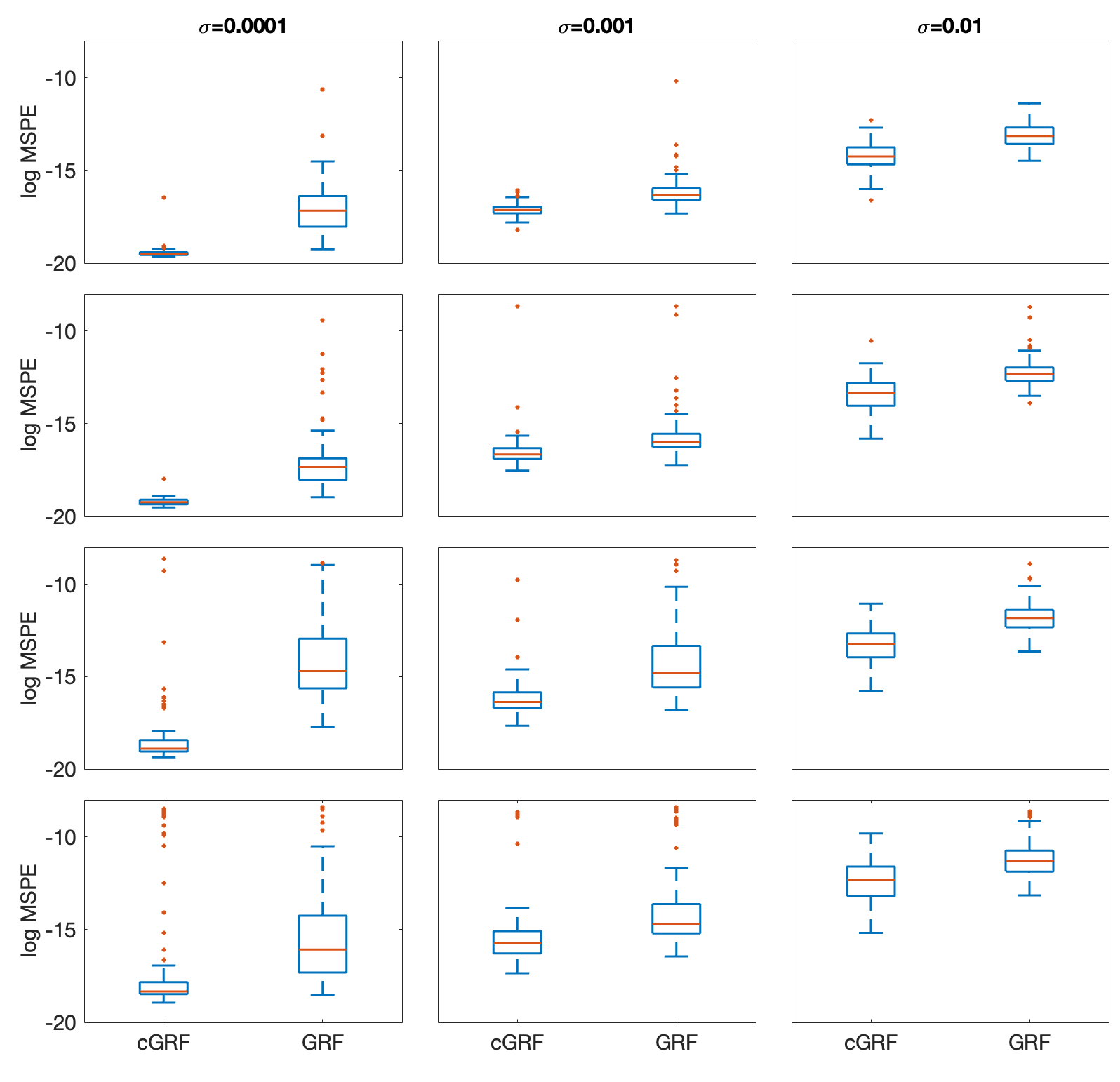} 
    \includegraphics[width=0.15\linewidth]{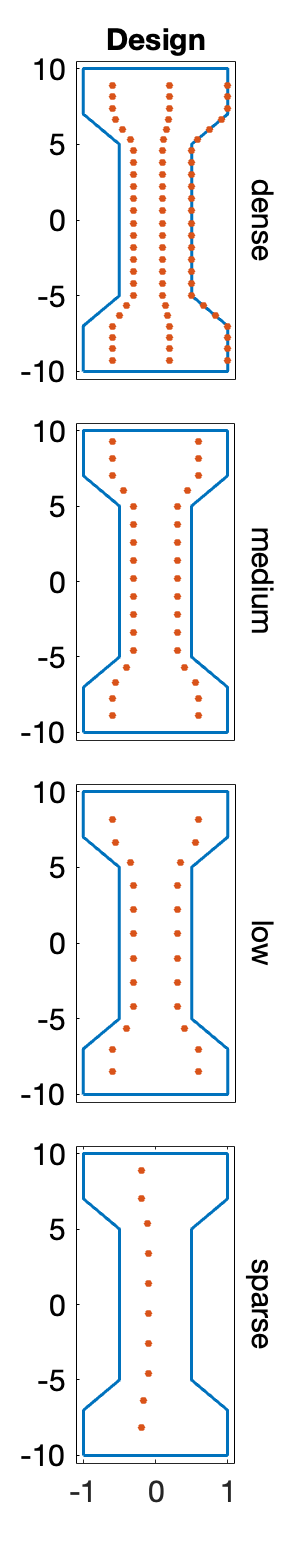}
    \caption{Box plots of log MSPE for displacement based on $100$ simulations comparing the effect of constrained and unconstrained GRF priors under different settings. Measurement error standard deviations are $\sigma=0.0001,0.001$ and $0.01$ from left to right. Designs are dense, medium, low and sparse from top to bottom. The right column shows the measurement spatial locations.}
    \label{fig:tensile}
\end{figure}

Let $y_i=u(t_i, \boldsymbol{x}_i) + \epsilon_i$, $i = 1,\ldots, N$ be noisy measurements of displacement with $\epsilon_i \overset{ind}{\sim} N(0,\sigma^2)$. We consider two prior models for $u$: an unconstrained GRF with $0$ mean and squared exponential covariance $k_0$, and a cGRF constructed from a base GRF with the same mean and covariance, i.e., $cGRF(\boldsymbol{\mathcal{L}},\boldsymbol{g},\boldsymbol{A}, 0, k_0)$ with $\boldsymbol{\mathcal{L}}=\{\mathcal{I},\mathcal{I}\}$, and $\boldsymbol{g}=\{0,0.005t\}$. The dataset is split into training, validation and testing data, at times $t\in \{0,0.3,0.6,0.9\}$, $t\in\{0.1,0.4,0.7,1\}$ and $t\in\{0.2,0.5,0.8\}$, respectively. Discrete measurements in the spatial domain are collected according to four designs shown in the right column of Figure \ref{fig:tensile}. We select optimal hyperparameters for both the $cGRF$ and its unconstrained counterpart by minimizing mean square prediction error (MSPE) based on the validation data.
Predictive performance is evaluated at the test locations via the MSPE under four different designs and three different noise levels. As shown in Figure \ref{fig:tensile}, boundary-constrained state estimation via cGRF outperforms the unconstrained case in all scenarios. Notice that outliers with large log MSPE, represented by orange dots in Figure \ref{fig:tensile}, typically occur when the numerical optimizer for hyperparameter tuning converges to local minima, corresponding to small values of length-scale in the $x_2$ dimension.

\subsection{Real data: displacement prediction in tensile testing}
\label{sec:4.4}

We return to the problem of strain field estimation, this time based on experimental data. The tensile testing experiment was conducted on a polycarbonate dog-bone-shaped specimen prepared according to ASTM D638, an industry standard for determining tensile properties. The experimental setting is the same as in Section \ref{sec:4.3}: the lower boundary of the specimen is fixed, while the upper boundary is pulled upwards by the moving crosshead at a constant speed. During the experiment, displacement measurements are obtained using digital image correlation (DIC), a non-contact optical technique that tracks the motion of speckle patterns over the surface of the specimen across sequential images during deformation. The prior on the latent displacement field is constructed using the cGRF framework and is updated by conditioning on discretely-observed DIC measurements. Due to instability at the beginning of the experiment, we take time step $100$ as the reference configuration and model the relative vertical displacement with respect to that image over subsequent time steps. Consider the spatio-temporal domain $\mathcal{D}=[100,200]\times \mathcal{D}_s$, where $\mathcal{D}_s$ is the spatial domain with lower boundary $A_1=\{(x_1,x_2)\mid x_2=-37.5\}$ and upper boundary $A_2=\{(x_1,x_2)\mid x_2=35\}$. Let $u:\mathcal{D}\to\mathbb{R}$ denote the vertical displacement. Similar to the simulated setting, the boundary constraints are $u(t,x_1,x_2)=0.005t$ at $x_2=35$ and $u(t,x_1,x_2)=0$ at $x_2=-37.5$.

To compare the predictive performance between unconstrained and constrained GRFs, we randomly split the dataset into training, validation, and testing sets. Hyperparameters are selected by minimizing validation log MSPE, and predictive performance is evaluated on held-out testing data. Hyperparameter optimization is known to be challenging in GRF models \citep{Berger2001Objective,Karvonen2023Maximum}, particularly due to confounding of covariance hyperparameters and measurement error variance. Moreover, as observed in Figure \ref{fig:tensile}, the numerical optimizer occasionally converges to local minima. This issue becomes more pronounced for the real dataset. To reduce sensitivity to the choice of data split and the resulting hyperparameter optimization, we repeat the analysis over $100$ random splits. Furthermore, instead of jointly optimizing the length-scales and the error variance, we fix the measurement error variance at $\sigma^2\in\{10^{-8},10^{-7},10^{-6},10^{-5},10^{-4}\}$ and optimize only the length-scale hyperparameters. The range of the selected error variances is chosen based on domain-specific knowledge.

Figure \ref{fig:tensile_real} shows the mean log MSPE together with $\pm2$ standard deviation intervals across the $100$ random splits for both the cGRF and GRF models. The cGRF consistently achieves lower average log MSPE and reduced variability across most measurement error levels. These results suggest that incorporating known physical boundary mechanisms through the cGRF framework improves prediction of displacement fields in real tensile testing applications.

\begin{figure}[t]
    \centering
    \includegraphics[width=1\linewidth]{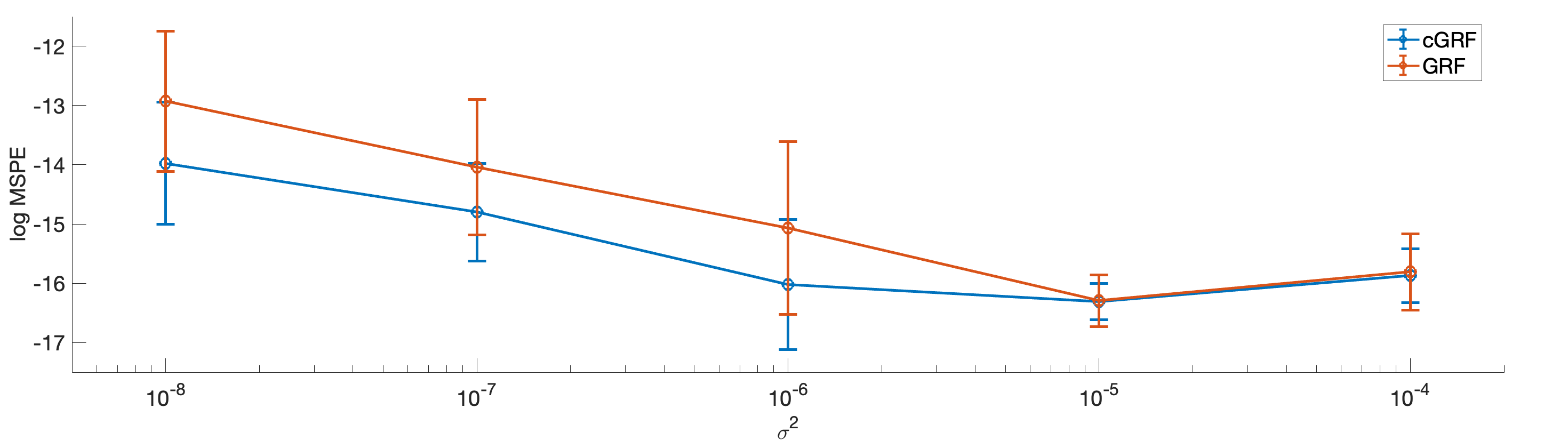} 
    
    \caption{Visualization of log MSPE for displacement based on $100$ random splits comparing the effect of constrained and unconstrained GRF priors, both optimized for multiple measurement error standard deviations.}
    \label{fig:tensile_real}
\end{figure}

\section{Conclusion}\label{sec-conc}

An important challenge in working with Gaussian random field models in physical, environmental, and engineering applications is the difficulty of directly enforcing linear boundary constraints on the states, while retaining their modeling and computational advantages. We developed a novel framework to construct constrained GRFs over multi-dimensional domains by transforming so-called base GRFs, and inheriting their smoothness properties. Importantly, our definition of a boundary-constrained process does not require establishing the existence of conditioned measures to constrain states over continuous boundaries. We showcased the new cGRF model in a variety of applications where boundary constraints add important physics-based information to the inference procedure and lead to performance improvement over their unconstrained counterparts. Applications of this framework are wide-ranging and include any problem in which the user has prior knowledge about the state's behavior at the boundaries.

Our case study highlights the potential for impact of this method in material science engineering where high fidelity reduced order models and simulation are an important growth area for the certification and qualification of additively manufactured components. We explored the example of tensile tests where one typically knows the displacements at the grips (e.g., a fixed lower grip and a prescribed upward motion at the upper grip) in advance. Incorporating these fixed boundary conditions into a cGRF-based inference framework ensures that the predicted displacement field exactly satisfies the physical constraints at the specimen ends, resulting in more accurate full-field predictions and more realistic uncertainty quantification. This improved inference propagates directly to the stress–strain curve and the estimation of material properties, such as Young’s modulus and failure strain, which depend sensitively on the accuracy of the strain measurements derived from the displacement field. Such enhancements are particularly crucial for additive manufacturing processes, where materials often exhibit anisotropic, process-dependent mechanical behavior \citep{Cai2025-ln}. By capturing the true mechanical response with higher fidelity, the cGRF approach supports the qualification of additively manufactured parts for demanding, high-performance applications while potentially reducing the extent of destructive testing needed for validation.

Other modeling extensions include using cGRFs with fixed-state boundary constraints as local models for domain partitions to ensure state continuity. By partitioning the domain into regions each modeled by a stationary GRF, \citet{Gramacy2008TreedGP, Gramacy2015LocalGP} and \citet{Huang2025StickGP} model global nonstationarity in the data, and enable distributed-memory computation in today's high-performance-computing environments. A drawback is that this induces discontinuities in the underlying smooth states. In such cases, cGRF priors may be joined at the boundaries to enforce state continuity, yielding more physically realistic models. Moreover, extensions to enforce state differentiability at the partition joints are possible with cGRFs constrained on linear combinations of the state and its derivatives. 
Second, Bayesian data-driven discovery of dynamical systems is an important open area \citep{North2025DiscoveryPDE}, and cGRF models of system states enable GRF-assisted algorithms \citep{Chen2020Lulu}.  Similarly, using boundary-constrained priors in the probabilistic solution of PDEs enables advances such as the modeling of epistemic uncertainty across solvers with multiple grid resolutions, directly characterizing the trade-off between solver fidelity and accuracy. In summary, the broader impact of using cGRF models to enforce continuous constraints on states is in enabling further progress towards more realistic physics-informed uncertainty quantification methods.

While the cGRF approach provides a new and flexible modeling tool, an investigation of the theoretical properties under different projection and weight function choices would provide practitioners with a better understanding of their expressivity and optimality. Another important open question is how to construct priors on states with nonlinear boundary constraints. The cGRF approach fails in this case because GRFs are not closed under nonlinear transformations. However, since Definition \ref{def:constrGRF} applies to general stochastic processes, a non-Gaussian analogue or approximation may be feasible.

\section{Disclosure statement}\label{disclosure-statement}

No conflicts of interest exist, to the best of our knowledge.

\section{Data Availability Statement}\label{data-availability-statement}

Experimental tensile testing data is provided in the project repository under the URL 
\href{https://github.com/amyymaa/cGRF/tree/main/data}{https://github.com/amyymaa/cGRF/tree/main/data}.

\section{Acknowledgment}

The authors thank Artur Leonel Machado Ulsenheimer and Calvin M. Stewart for conducting the tensile testing experiment and providing the dataset analyzed in Section \ref{sec:4.4}. Their experimental support and domain expertise were invaluable in evaluating the proposed methodology in this engineering application.

\bibliography{bibliography}

\end{document}


\def\spacingset#1{\renewcommand{\baselinestretch}%
{#1}\small\normalsize} \spacingset{1}

\if1\anon
{
  \title{\bf Supplement for ``Constrained Gaussian Random Fields with Continuous Linear Boundary Restrictions for Physics-informed Modeling of States''}
  \author{Yue Ma\thanks{The authors would like to acknowledge support from the U.S. National Science Foundation Engineering Research Center for Hybrid Autonomous Manufacturing Moving from Evolution to Revolution (ERC‐HAMMER) under Award Number EEC-2133630.}
    \hspace{.2cm}\\
    \normalsize Department of Statistics, The Ohio State University\\
    Oksana Chkrebtii \\
    \normalsize Department of Statistics, The Ohio State University\\
    Steve Niezgoda \\
    \normalsize Department of Materials Science and Engineering, The Ohio State University}
  \maketitle
} \fi

\if0\anon
{
  \title{\bf Supplement for ``Constrained Gaussian Random Fields with Continuous Linear Boundary Restrictions for Physics-informed Modeling of States''}
} \fi

\spacingset{1.8}

\section{Theoretical justifications}
\subsection{Proof of Theorem 1}

Before proving Theorem 1, we expand Definition 1 to a random field $u^{\boldsymbol{A}}$ which is linearly constrained to satisfy \emph{multiple} linear boundary constraints.

\setcounter{definition}{1}
\begin{definition}
    \label{def:constrGRF1}
    For $i=1,...,n$, let $\mathcal{L}_i:\mathcal{H}\to \mathcal{H}_i'$ be a linear operator, $g_i:\mathcal{D}\to \mathbb{R}$ be a \textit{target} function, $A_i\subset \partial\mathcal{D}$ be a segment of the boundary, and 
    $u^{\boldsymbol{A}}: \mathcal{D}\times \Omega \to \mathbb{R}$
    be a random field. We say $u^{\boldsymbol{A}}$ is linearly-constrained if for $i=1,...,n$, $\mathcal{L}_iu^{\boldsymbol{A}}$ is sample-continuous and $\mathcal{L}_iu^{\boldsymbol{A}}(\boldsymbol{x})$ equals $g_i(\boldsymbol{x})$ for all $\boldsymbol{x}\in A_i$ almost surely.
\end{definition}

The proof of Theorem 1 is as follows.

\begin{proof}
    For $i=1,...,n$ and $\boldsymbol{x}\in \mathcal{D}$, with probability one,
    \begin{align*}
        \mathcal{L}_iu^{\boldsymbol{A}}(\boldsymbol{x}) &= \mathcal{L}_iu(\boldsymbol{x}) + \sum_{j=1}^n\mathcal{L}_i \left[ w_j(\boldsymbol{x})(g_j\circ f_j(\boldsymbol{x}) - \mathcal{L}_ju\circ f_j(\boldsymbol{x}))\right] \\
        &= \mathcal{L}_iu(\boldsymbol{x}) + g_i\circ f_i(\boldsymbol{x}) - \mathcal{L}_iu\circ f_i(\boldsymbol{x}),
    \end{align*}
    where the second equality follows from condition (4) of Theorem 1. Next, to show that $\mathcal{L}_iu^{\boldsymbol{A}}$ is sample-continuous, note that  $\mathcal{L}_iu\in \mathcal{H}_i'$ and $g_i\circ f_i - \mathcal{L}_iu\circ f_i \in \mathcal{H}_i'$ imply that  $\mathcal{L}_iu^{\boldsymbol{A}} \in \mathcal{H}_i' \subset \mathcal{C}(\mathcal{D})$. Since for all $\boldsymbol{x}\in A_i$
    \begin{align*}
        \mathcal{L}_iu^{\boldsymbol{A}}(\boldsymbol{x}) &= \mathcal{L}_iu(\boldsymbol{x}) + g_i(\boldsymbol{x}) - \mathcal{L}_iu(\boldsymbol{x})
        = g_i(\boldsymbol{x}),
    \end{align*}
    almost surely, $u^{\boldsymbol{A}}$ is linearly-constrained by Definition \ref{def:constrGRF1}. Furthermore, since $u^{\boldsymbol{A}}$ is a linear transformation of a GRF, then it is also a GRF. 

\end{proof}

\subsection{Relationship with existing methods}

In Example $1$, we construct cGRFs on the interval $\mathcal{D}=[0,1]$. We first consider enforcing fixed-state constraint $u(x)=0$ at $x=0$. Two such cGRFs are $u^{\{0\}}(x)=u(x)-u(0)$ and $u^{\{0\}}(x)=u(x)- (k(x,0)/k(0,0))u(0)$, corresponding to $w(x)=1$ and $w(x)=k(x,0)/k(0,0)$, respectively. 

Assuming $u$ is sample-differentiable, the first representation has the same covariance as the integral-based method in \citet{Chkrebtii2013thesis, Ye2022MultiFidelity}, where 
$$k^{\{0\}}(x,x') = \int_0^x\int_0^{x'}\partial_z\partial_{z'}k_0(z,z') \, dzdz' = k_0(x,x')-k_0(x,0)-k_0(0,x')+k_0(0,0).$$

Assuming $u(1)=0$, the second representation is the Gaussian bridge in Proposition $1$ of \citet{Gasbarra2007GB} with $\xi=\theta=0$ and $T=1$. 

We then consider enforcing fixed-derivative constraint $\partial_x u(x) = 0$ at $x = 0$. A cGRF is $u^{\{0\}}(x) = u(x) - x \ \partial_xu(0)$, corresponding to $w(x)=x$. For this particular case, the cGRF is equivalent to Eq. $(27)$ in \cite{Dalton2024BoundaryConstrainedGP} with $a(x) = h(x) = \hat{u}_2(x)=0$, and $w(x)$ in the cGRF is equal to $\phi(x)\nabla\phi(x)$, where $\phi$ is the normalized ADF from \cite{Dalton2024BoundaryConstrainedGP}.

\subsection{GRF smoothness and the linear operator}

Table \ref{tab:ex_cov} provides examples of appropriate base covariance functions $k_0$ for parametrizing a cGRF by constraint type. In case of mixed constraints, the smoothest covariance structure should be chosen. Note also that the base mean $m_0$ must be at least as smooth as the covariance.

\begin{table}[htbp]
\centering
\caption{Examples of sufficiently smooth base GRF covariances by constraint type}
\begin{tabular}{|l|l|}
\hline
\textbf{Constraint type} & \textbf{Examples of covariance structure} \\
\hline
Integral constraint & Exponential, squared exponential, Mat\'ern \\ \hline
State constraint & Squared exponential, Mat\'ern with $\nu > 1$ \\ \hline
Maximum derivative order 1      & Squared exponential, Mat\'ern with $\nu > 2$      \\ \hline
Maximum derivative order 2     & Squared exponential, Mat\'ern with $\nu > 3$              \\ \hline
Maximum derivative order 3            & Squared exponential, Mat\'ern with $\nu > 4$            \\ \hline
Maximum derivative order 4         & Squared exponential, Mat\'ern with $\nu > 5$         \\ \hline
\end{tabular}
\label{tab:ex_cov}
\end{table}

\section{Implementation details}
\subsection{Example: cGRF on a triangular domain}

Consider enforcing the fixed-state boundary constraints along two segments of a unit triangular domain $\mathcal{D}=\{(x_1,x_2)\in\mathbb{R}^2 \mid x_1,x_2\geq 0, x_1+x_2\leq1\}$ with boundary $\partial \mathcal{D}$, as shown in the left panel of Figure $7$. 
We wish to construct $u^{\boldsymbol{A}}\sim cGRF(\boldsymbol{\mathcal{L}}, \boldsymbol{g}, \boldsymbol{A}, m_0,k_0)$, where $\boldsymbol{\mathcal{L}} = \{\mathcal{I}, \mathcal{I}\}$, $\boldsymbol{g} = \{0, 0\}$, and $\boldsymbol{A} = \{A_1=\{\boldsymbol{x}\in \partial \mathcal{D} \mid x_1=0\}, A_2=\{\boldsymbol{x}\in \partial\mathcal{D} \mid x_1+x_2=1\}\}$, and $\boldsymbol{x}=(x_1,x_2)$. $m_0$ and $k_0$ are mean and covariance functions for a sample-continuous unconstrained GRF. For example, set $m_0=0$ and $k_0$ to be a squared exponential covariance. As discussed in the manuscript, the representation for such a cGRF is not unique, and here we demonstrate the use of the recipe for $\boldsymbol{w}$ in Section $3.3$, and two projection functions, $f_1(\boldsymbol{x})=(0,x_2)$ and $f_2(\boldsymbol{x})=(1-x_2,x_2)$, following directions $(-1,0)$ and $(1,0)$, respectively. For $\boldsymbol{x}\in \mathcal{D}$, $\boldsymbol{w}(\boldsymbol{x}) = \boldsymbol{v}(\boldsymbol{x})M^{-1}(\boldsymbol{x})$, where 
{\small
\begin{align*}
\begin{split}
    \boldsymbol{v}(\boldsymbol{x}) &= \begin{bmatrix}
    k_0(\boldsymbol{x}, (0,x_2)) & k_0(\boldsymbol{x}, (1-x_2,x_2))
    \end{bmatrix}, \\
    M(\boldsymbol{x}) &= 
    \begin{bmatrix}
    k_0((0,x_2),(0,x_2)) & k_0((0,x_2),(1-x_2,x_2)) \\
    k_0((1-x_2,x_2),(0,x_2)) & k_0((1-x_2,x_2),(1-x_2,x_2))
    \end{bmatrix}.
\end{split}
\end{align*}}
Therefore, $\boldsymbol{w}(\boldsymbol{x})=[w_1(\boldsymbol{x}) \ w_2(\boldsymbol{x})]^\top$, where
{\small
\begin{align*}
    w_1(\boldsymbol{x}) &=
    \frac{k_0(\boldsymbol{x}, (0,x_2))k_0((1-x_2,x_2),(1-x_2,x_2)) -k_0(\boldsymbol{x}, (1-x_2,x_2))k_0((1-x_2,x_2),(0,x_2))}{k_0((0,x_2),(0,x_2))k_0((1-x_2,x_2),(1-x_2,x_2)) - k_0((0,x_2),(1-x_2,x_2))^2}, \\
    w_2(\boldsymbol{x}) &= 
    \frac{k_0(\boldsymbol{x}, (1-x_2,x_2))k_0((0,x_2),(0,x_2)) -k_0(\boldsymbol{x}, (0,x_2))k_0((0,x_2),(1-x_2,x_2))}{k_0((0,x_2),(0,x_2))k_0((1-x_2,x_2),(1-x_2,x_2)) - k_0((0,x_2),(1-x_2,x_2))^2}.
\end{align*}
}%
Notice for all $\boldsymbol{x}\in A_1$, i.e., $\boldsymbol{x}=(0,x_2)$ for $x_2\in [0,1]$, $w_1(\boldsymbol{x}) = 1$ and $w_2(\boldsymbol{x}) = 0$. Similarly, for all $\boldsymbol{x}\in A_2$, i.e., $\boldsymbol{x}=(1-x_2,x_2)$ for $x_2\in [0,1]$, $w_1(\boldsymbol{x}) = 0$ and $w_2(\boldsymbol{x}) = 1$. In other words, condition (4) is satisfied using these $\boldsymbol{w}$. Based on Theorem 1,
{\small
$$u^{\boldsymbol{A}}(\boldsymbol{x}) = u(\boldsymbol{x}) - w_1(\boldsymbol{x})u(0,x_2) - w_2(\boldsymbol{x})u(1-x_2,x_2).$$}
Based on Lemma 1,
{\small
\begin{align*}
    m^{\boldsymbol{A}}(\boldsymbol{x}) &= 0, \\
    k^{\boldsymbol{A}}(\boldsymbol{x}, \boldsymbol{x}') &= k_0(\boldsymbol{x}, \boldsymbol{x}') - w_1(\boldsymbol{x}')k_0(\boldsymbol{x},(0,x_2')) - w_2(\boldsymbol{x}')k_0(\boldsymbol{x},(1-x_2',x_2')) \\
    &- w_1(\boldsymbol{x})k_0((0,x_2), \boldsymbol{x}') - w_2(\boldsymbol{x})k_0((1-x_2,x_2), \boldsymbol{x}') + w_1(\boldsymbol{x}')w_1(\boldsymbol{x})k_0((0,x_2),(0,x_2')) \\
    &+ w_1(\boldsymbol{x}')w_2(\boldsymbol{x})k_0((1-x_2,x_2),(0,x_2')) + w_2(\boldsymbol{x}')w_1(\boldsymbol{x})k_0((0,x_2),(1-x_2',x_2')) \\
    &+ w_2(\boldsymbol{x}')w_2(\boldsymbol{x})k_0((1-x_2,x_2),(1-x_2',x_2')).
\end{align*}
}%
Other types of cGRFs can be constructed following a similar procedure. We also provide MATLAB code for all the cGRFs constructed in the manuscript.

\subsection{Example: cGRF on a disk-shaped domain}

Another example is to construct a GRF $u^{\partial\mathcal{D}}$ constrained to $(x_1+x_2)^2$ at the circular boundary $\partial\mathcal{D}$ of the unit disk $\mathcal{D}=\{\boldsymbol{x}=(x_1,x_2)\in \mathbb R^2 \mid x_1^2 + x_2^2 \leq 1\}$. The cGRF is
$$u^{\partial\mathcal{D}}(\boldsymbol{x}) = u(\boldsymbol{x}) - (1-w(\boldsymbol{x})) \ u(-(1-x_2^2)^{1/2}, x_2) - w(\boldsymbol{x}) \ u((1-x_2^2)^{1/2}, x_2) + (x_1+x_2)^2,$$
where $w(\boldsymbol{x})=(x_1+(1-x_2^2)^{1/2}) / (2(1-x_2^2)^{1/2})$. 
Mean and covariance functions for $u^{\partial\mathcal{D}}$ are
{\small
\begin{align*}
    m^{\partial\mathcal{D}}(\boldsymbol{x}) &= (x_1+x_2)^2, \\
    k^{\partial\mathcal{D}}(\boldsymbol{x}, \boldsymbol{x}') &= k_0(\boldsymbol{x}, \boldsymbol{x}') - (1-w(\boldsymbol{x}'))k_0(\boldsymbol{x},(-(1-x_2'^2)^{1/2}, x_2')) - w(\boldsymbol{x}')k_0(\boldsymbol{x},((1-x_2'^2)^{1/2}, x_2')) \\
    &- (1-w(\boldsymbol{x}))k_0((-(1-x_2^2)^{1/2}, x_2), \boldsymbol{x}') - w(\boldsymbol{x})k_0(((1-x_2^2)^{1/2}, x_2), \boldsymbol{x}') \\
    &+ (1-w(\boldsymbol{x}'))(1-w(\boldsymbol{x}))k_0((-(1-x_2^2)^{1/2}, x_2),(-(1-x_2'^2)^{1/2}, x_2')) \\
    &+ (1-w(\boldsymbol{x}'))w(\boldsymbol{x})k_0(((1-x_2^2)^{1/2}, x_2),(-(1-x_2'^2)^{1/2}, x_2')) \\
    &+ w(\boldsymbol{x}')(1-w(\boldsymbol{x}))k_0((-(1-x_2^2)^{1/2}, x_2),((1-x_2'^2)^{1/2}, x_2')) \\
    &+ w(\boldsymbol{x}')w(\boldsymbol{x})k_0(((1-x_2^2)^{1/2}, x_2),((1-x_2'^2)^{1/2}, x_2')).
\end{align*}
}%

\subsection{Example: cGRF on an annulus domain}

The right panel of Figure $2$ shows samples from a GRF $u^{\partial\mathcal{D}}$ constrained to $0$ at the concentric circular boundaries $\partial\mathcal{D}$ of the annulus domain $\mathcal{D}=\{\boldsymbol{x}=(x_1,x_2)\in \mathbb R^2 \mid 0.2^2\leq x_1^2 + x_2^2 \leq 1\}$. Unlike the examples above, here we illustrate the use of radial projections. The cGRF is
$$u^{\partial\mathcal{D}}(\boldsymbol{x}) = u(\boldsymbol{x}) - w_1(\boldsymbol{x}) \ u(f_1(\boldsymbol{x})) - w_2(\boldsymbol{x}) \ u(f_2(\boldsymbol{x})).$$
Here $w_1(\boldsymbol{x})$ and $w_2(\boldsymbol{x})$ are constructed using equation $(6)$ in Conjecture $1$, where $\boldsymbol{v}_i(\boldsymbol{x}) = k_0(\boldsymbol{x}, f_i(\boldsymbol{x}))$ and $M_{ij}(\boldsymbol{x})=k_0(f_i(\boldsymbol{x}), f_j(\boldsymbol{x}))$, for $i,j=1,2$. The two radial projections are $f_1(\boldsymbol{x})=0.2\boldsymbol{x}/||\boldsymbol{x}||$ and $f_2(\boldsymbol{x})=\boldsymbol{x}/||\boldsymbol{x}||$, corresponding to the projections to the inner and outer circular boundaries, respectively. Mean and covariance functions for $u^{\partial\mathcal{D}}$ can be derived from Lemma $1$ accordingly.

\subsection{Example: cGRF with product covariance structure}

Consider a base GRF $u$ on the unit square domain $\mathcal{D}=[0,1]^2$ with mean function $0$ and covariance function 
$$k_0(\boldsymbol{x}, \boldsymbol{x}') = k_1(x_1,x_1')k_2(x_2,x_2').$$
Like in the right panel of Figure $7$, we wish to enforce two-sided parallel boundary constraints on $A=\{(x_1,x_2)\in \partial\mathcal{D} \mid x_2=0 \text{ or } x_2=1\}$. The representation is $u^{A}(x_1,x_2)=u(x_1,x_2)-(1-x_2)u(x_1,0) - x_2u(x_1,1)$, for $(x_1,x_2) \in \mathcal{D}$, with mean function $m^{A}(\boldsymbol{x}) = 0$ and covariance function 
{\small
\begin{align*}
    k^{A}(\boldsymbol{x}, \boldsymbol{x}') &= k_0(\boldsymbol{x}, \boldsymbol{x}') - (1-x_2')k_0(\boldsymbol{x},(x_1',0)) - x_2'k_0(\boldsymbol{x},(x_1',1)) \\
    &- (1-x_2)k_0((x_1,0), \boldsymbol{x}') - x_2k_0((x_1,1), \boldsymbol{x}') + (1-x_2')(1-x_2)k_0((x_1,0),(x_1',0)) \\
    &+ (1-x_2')x_2k_0((x_1,1),(x_1',0)) + x_2'(1-x_2)k_0((x_1,0),(x_1',1)) + x_2'x_2k_0((x_1,1),(x_1',1)) \\
    &= k_1(x_1,x_1') [k_2(x_2,x_2') - (1-x_2')k_2(x_2,0) - x_2'k_2(x_2,1) - (1-x_2)k_2(0, x_2') - x_2k_2(1, x_2') \\
    &+ (1-x_2')(1-x_2)k_2(0,0) + (1-x_2')x_2k_2(1,0) + x_2'(1-x_2)k_2(0,1) + x_2'x_2k_2(1,1)] \\
    &=: k_1(x_1,x_1')k_2^A(x_2,x_2').
\end{align*}
}%
Note that $k_2^{A}$ is the covariance of a cGRF on the domain $[0,1]$ with state constraints at both endpoints, and the product covariance structure is preserved after the transformation. To reproduce Figure $7$, we set $k_1$ to be the periodic covariance in dimension $x_1$ and  Mat\'ern in $x_2$.

\subsection{Example: cGRF on non-parametrized domain}

The right panel of Figure $5$ shows samples from a cGRF $u^{\partial\mathcal{D}}$ constrained to $0$ on the boundary $\partial\mathcal{D}$ of an arbitrarily-shaped domain $\mathcal{D}$, County Laois in Ireland. We use a piecewise linear function with $11$ line segments to approximate the original boundary whose parameterization is unavailable. Let $\boldsymbol{x}=(x_1,x_2)\in\mathcal{D}$. The mean and covariance functions, $m_0$ and $k_0$, for the base GRF are $0$ and Mat\'ern with $\lambda=0.2$ and $\nu=2.5$, respectively. Below we define the projection and weight functions and the rest follows from Theorem $1$ and Lemma $1$.
For convenience, we use Cartesian projections in four directions $(0,1), (0,-1), (-1,0), (1,0)$. The expressions are
$
f_1(\boldsymbol{x}) = (x_1,b_T(x_1)), \ 
f_2(\boldsymbol{x}) = (x_1,b_B(x_1)), \
f_3(\boldsymbol{x}) = (b_L(x_2),x_2), \ \text{and }
f_4(\boldsymbol{x}) = (b_R(x_2),x_2) \
$, respectively,
where
{\small
\[
\begin{aligned}
b_T(x_1)&=53.17, \\
b_B(x_1)
&=
\left\{52.78+\frac{52.89-52.78}{-7.20+7.67}(x_1+7.67)\right\}
\mathbbm{1}\{x_1\le -7.20\} \\
&\quad+
\left\{52.89+\frac{52.81-52.89}{-6.97+7.20}(x_1+7.20)\right\}
\mathbbm{1}\{x_1>-7.20\}, \\
b_L(x_2)
&=
\left\{-7.63+\frac{-7.59+7.63}{53.10-53.17}(x_2-53.17)\right\}
\mathbbm{1}\{x_2\ge 53.10\} \\
&\quad+
\left\{-7.59+\frac{-7.71+7.59}{53.01-53.10}(x_2-53.10)\right\}
\mathbbm{1}\{53.01\le x_2<53.10\} \\
&\quad+
\left\{-7.71+\frac{-7.66+7.71}{52.93-53.01}(x_2-53.01)\right\}
\mathbbm{1}\{52.93\le x_2<53.01\} \\
&\quad+
\left\{-7.66+\frac{-7.73+7.66}{52.86-52.93}(x_2-52.93)\right\}
\mathbbm{1}\{52.86\le x_2<52.93\} \\
&\quad+
\left\{-7.73+\frac{-7.67+7.73}{52.78-52.86}(x_2-52.86)\right\}
\mathbbm{1}\{x_2<52.86\},\\
b_R(x_2)
&=
\left\{-7.07+\frac{-7.09+7.07}{53.00-53.17}(x_2-53.17)\right\}
\mathbbm{1}\{x_2\ge 53.00\} \\
&\quad+
\left\{-7.09+\frac{-6.95+7.09}{52.95-53.00}(x_2-53.00)\right\}
\mathbbm{1}\{52.95\le x_2<53.00\} \\
&\quad+
\left\{-6.95+\frac{-6.97+6.95}{52.81-52.95}(x_2-52.95)\right\}
\mathbbm{1}\{x_2<52.95\}.
\end{aligned}
\]
}%
The weight functions are defined according to Equation $(6)$ in Conjecture $1$, where $\boldsymbol{v}_i(\boldsymbol{x}) = k_0(\boldsymbol{x}, f_i(\boldsymbol{x}))$ and $M_{ij}(\boldsymbol{x})=k_0(f_i(\boldsymbol{x}), f_j(\boldsymbol{x}))$, for $i,j=1,...,4$.

\subsection{Additional application results}

\subsubsection{Data-driven discovery of dynamical systems}

\setcounter{figure}{10}
\begin{figure}
    \centering
    \includegraphics[width=1\linewidth]{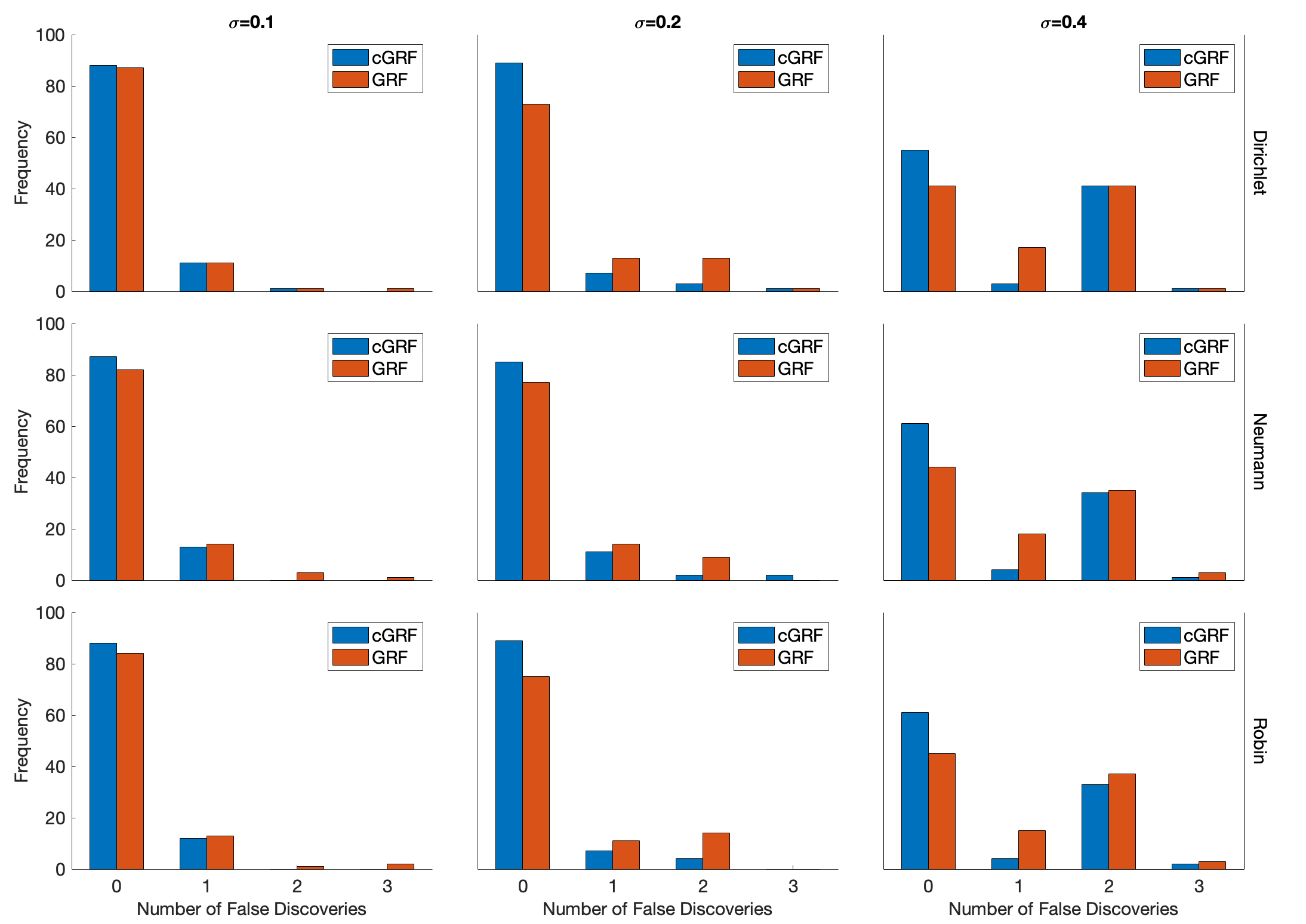}
    \caption{Frequencies of false discoveries in $100$ simulations for $cGRF$ in blue and $GRF$ in orange, under different settings. Columns $1-3$: measurement error $\sigma=0.1,0.2$ and $0.4$, respectively. Rows $1-3$: impose two-sided Dirichlet, Neumann and Robin boundary constraints, respectively.}
    \label{fig:fnfp}
\end{figure}

Figure \ref{fig:fnfp} provides additional results comparing false discovery rates using cGRFs versus their unconstrained versions.

\bibliography{bibliography.bib}